%% file: paper.tex
\documentclass[sigconf,final]{acmart}

\input{mydefs}

\usepackage{tikz}
\usepackage{pgfplots}
\usepackage{xcolor}
\usepackage{tabularx}

\usepackage{graphicx}
\usepackage{enumitem}
\setlist[itemize]{leftmargin=*}
\setlist[enumerate]{leftmargin=*}

\begin{document}

\title{Accountable Fine-grained Blockchain Rewriting in the Permissionless Setting}

\author{Yangguang Tian}
\affiliation{%
  \institution{Osaka University}
  \country{Japan}
}
\email{}

\author{Bowen Liu}
\affiliation{%
  \institution{Singapore University of \\Technology and Design}
  \country{Singapore}
}
\email{bowen_liu@mymail.sutd.edu.sg}

\author{Yingjiu Li}
\affiliation{%
  \institution{University of Oregon}
  \country{USA}
}
\email{yingjiul@uoregon.edu}

\author{Pawel Szalachowski}
\affiliation{%
  \institution{Singapore University of \\Technology and Design}
  \country{Singapore}
}
\email{pjszal@gmail.com}

\author{Jianying Zhou}
\affiliation{%
  \institution{Singapore University of \\Technology and Design}
  \country{Singapore}
}
\email{jianying_zhou@sutd.edu.sg}

\renewcommand{\shortauthors}{Tian and Liu, et al.}

\begin{abstract}

Blockchain rewriting with fine-grained access control allows a user to create a transaction associated with a set of attributes, while another user (or modifier) who possesses enough rewriting privileges from a trusted authority satisfying the attribute set can rewrite the transaction. However, it lacks accountability and is not designed for open blockchains that require no trust assumptions. In this work, we introduce accountable fine-grained blockchain rewriting in a permissionless setting. The property of accountability allows the modifier's identity and her rewriting privileges to be held accountable for the modified transactions in case of malicious rewriting (e.g., modify the registered content from good to bad). We first present a generic framework to secure blockchain rewriting in the permissionless setting. Second, we present an instantiation of our approach and show its practicality through evaluation analysis. Last, we demonstrate that our proof-of-concept implementation can be effectively integrated into open blockchains. 

\end{abstract}


\maketitle

\section{Introduction}

Blockchains have received tremendous attention from research communities and industries in recent years. The concept was first introduced in the context of Bitcoin \cite{nakamoto2008bitcoin}, where all payment transactions are appended in a public ledger, and each transaction is ordered and verified by network nodes in a peer-to-peer manner. Blockchain ledgers grow by one block at a time, where the new block in the chain is decided by a consensus mechanism (e.g., Proof-of-Work in Bitcoin \cite{jakobsson1999proofs}) executed by the network nodes. Usually, blockchains deploy hash-chains as an append-only structure, where the hash of a block is linked to the next block in the chain. Each block includes a set of valid transactions which are accumulated into a single hash value using the Merkle tree \cite{merkle1989certified}, and each transaction contains certain content which needs to be registered in the blockchain. 

Blockchains are designed to be immutable, such that the registered content cannot be modified once they are appended. However, blockchain rewriting is often required in practice, or even legally necessary in data regulation laws such as GDPR in Europe \cite{GDPR}. Since the platform is open, it is possible some users append transactions into a chain containing illicit content such as sensitive information, stolen private keys, and inappropriate videos \cite{matzutt2016poster,matzutt2018quantitative}. The existence of illicit content in the chain could pose a challenge to law enforcement agencies like Interpol \cite{tziakouris2018cryptocurrencies}. 

Blockchain rewriting is usually realized by replacing a standard hash function, used for generating transaction hash in the blockchain, by a trapdoor-based chameleon hash \cite{DBLP:conf/ndss/KrawczykR00}. Then the users, who have the same privilege (i.e., hold the trapdoor), can modify a transaction. In other words, the same transaction can be modified by users with the same privileges only. Nonetheless, for most real-life blockchain applications, blockchain rewriting with fine-grained access control is desired so that  various rewriting privileges can be granted to different modifiers and that the same transaction can be modified by users with different privileges. For blockchain rewriting with fine-grained access control, a user first associates his transaction with a set of attributes. Then, any user possessing a chameleon trapdoor can modify the transaction if her access privilege corresponding to the chameleon trapdoor satisfies the embedded attribute set. It is possible that multiple users with different access privileges can modify a same transaction if their privileges satisfy the set of attributes associated with the transaction. 

\smallskip\noindent\textbf{Motivation.} Blockchain rewriting with fine-grained access control has been recently studied in the permissioned setting \cite{DBLP:conf/ndss/DerlerSSS19}; however, the proposed solution is not suitable for permissionless (i.e., open) blockchains for two reasons: 1) It requires a trusted authority to distribute rewriting privileges. 2) It lacks accountability as modifiers may rewrite the blockchain without being identified. The main motivation of this work is to extend blockchain rewriting with fine-grained access control to open blockchains such as Bitcoin \cite{bitcoin} and Ethereum \cite{ethereum}, which assume no trusted authority for managing access privileges.

Blockchain rewriting with fine-grained access control requires accountability especially in the permissionless setting. It is more critical for open blockchains, as there is no trusted authority to identify which modifier is responsible for any maliciously modified transaction (e.g., a modifier may rewrite the registered content in a transaction from good to bad). Besides, if modifiers attempt to generate an access device (or blackbox) that accumulates various kinds of rewriting privileges, and distribute or sell it to the public, it is more challenging to identify the responsible rewriting privileges given a modified transaction. This is because any user may rewrite a transaction successfully if he holds that access blackbox. Therefore, public accountability in this work encompasses two aspects: 1) identify a modifier in case of malicious modification on a transaction, and 2) identify the modifier's rewriting privileges even if an access blackbox is used. 

Let us consider an open blockchain system that includes multiple committees, and each committee contains multiple users appointed for granting rewriting privileges. For ease of exposition, we assume there exists only one committee and one modifier for modifying mutable transactions during an epoch (i.e., a fixed interval of time). We have two blockchain rewriting cases: 1) a modifier rewrites a transaction using her access privilege granted from a committee, and the modifier may have various access privileges by joining different committees at different epochs. 2) an unauthorized user has no rewriting privilege from any committee; he may still rewrite a transaction using an access blackbox that accumulates various modifiers' access privileges. In such cases, it is desired to hold the rewriting privileges accountable for the modified transactions. Overall, if a transaction is maliciously modified, any public user can identify the responsible rewriting privilege that links to either a designated modifier or an unauthorized user with an access blackbox. Public accountability allows the modifiers (including the unauthorized users with access blackboxes) and the rewriting privileges to be held accountable for the modified transactions in open blockchain systems.

\smallskip\noindent\textbf{This Work.} We introduce a new framework of accountable fine-grained blockchain rewriting, which is used to secure blockchain rewriting in a permissionless setting. The proposed framework achieves strong security without trusted authority and achieves public accountability simultaneously. Strong security without trusted authority indicates that the proposed framework remains secure even when attackers can compromise a threshold number of users in any committee. Public accountability means that any user in the public can link a modified transaction to a modifier and the responsible rewriting privilege.

We now explain our key technical insights. First, we rely on dynamic proactive secret sharing (DPSS) \cite{maram2019churp} to remove the trusted authority and achieve strong security. We replace the trusted authority by a committee of multiple users, where each user holds a share of the trust. We allow any user to join in and leave from a committee in any time epoch. Since we use key-policy attribute-based encryption (KP-ABE) \cite{goyal2006attribute} to ensure fine-grained access control, we highlight the following points: 1) The master secret key is split into multiple key shares so that each user in a committee holds a single key share. 2) A certain number of shareholders in a committee collaboratively recover the master secret key and distribute access privileges to modifiers. 3) Any user can freely join/leave a committee, and the master secret key remains fixed across different committees. We achieve strong security because the master secret key remains secure even if the attackers are allowed to compromise a threshold number of shareholders in any committee.

Second, we achieve public accountability using a digital signature scheme, a commitment scheme, and KP-ABE with public traceability (ABET). We show the purpose of using those primitives as follows: 1) The digital signature helps the public link a modified transaction to a modifier (or modifier's public key), as she signs the modified transaction using her signing key, and the signed transaction is publicly verifiable. 2) The commitment scheme helps the public link modifiers' public keys with committees. The intention is to show a modifier is indeed obtained a rewriting privilege from a committee. 3) If an unauthorized user holds an access blackbox, ABET helps the public to obtain a set of rewriting privileges from interacting with the access blackbox due to ABET's public traceability. Since there is no existing ABET can be applied to this work, we propose a new ABET scheme, which is suitable for decentralized systems such as open blockchains.

\smallskip\noindent\textbf{Our Contributions.} The major contributions of this work are summarized as follows.

\begin{itemize}

    \item{\it Generic Framework.} We introduce a new generic framework of accountable fine-grained blockchain rewriting, which is based on the chameleon hash function. A unique feature of this framework is that it allows the fine-grained blockchain rewriting to be performed in the permissionless setting. 
    
    \item{\it Public Accountability.} We introduce a new notion called public accountability, such that modifiers' public keys and rewriting privileges are held accountable for the modified transactions, which is essential to blockchain rewriting because it helps thwart malicious rewriting. 

  \item{\it Practical Instantiation.} We present a practical instantiation, and our evaluation analysis validates its practicality. We present an efficient ABET, which is of independent interest. The proposed ABET scheme is the first KP-ABE scheme with public traceability designed for decentralized systems. 
  
  \item{\it Integration to Open Blockchains.} The proof-of-concept implementation shows that blockchain rewriting based on our approach incurs almost no overhead to chain validation when compared to the immutable blockchain.
    
\end{itemize}

\section{Preliminary}

In this section, we present the complexity assumptions and the building blocks, which are used in our proposed generic construction and instantiation.

\subsection{Complexity Assumptions}
\label{assumptions}
\noindent\textbf{Bilinear Maps.} Let $(g, h)$ denote two group generators, which takes a security parameter $\lambda$ as input and outputs a description of a group $\mathbb{G}, \mathbb{H}$. We define the output of $(g , h)$ as $(q, \mathbb{G}, \mathbb{H} ,\mathbb{G}_T, \e)$, where $q$ is a prime number, $\mathbb{G}$, $\mathbb{H}$ and $\mathbb{G}_T$ are cyclic groups of order $q$, and $\e: \mathbb{G} \times \mathbb{H} \rightarrow \mathbb{G}_T$ is a bilinear map such that: (1) Bilinearity: $\forall g,h \in \mathbb{G}$ and $a,b \in \mathbb{Z}_q$, we have $\e (g^a,h^b)=e(g,h)^{ab}$; (2)  Non-degeneracy: $\exists g \in \mathbb{G}$ such that $\e(g,h)$ has order $q$ in $\mathbb{G}_T$. We assume that group operations in $\mathbb{G}$, $\mathbb{H}$ and $\mathbb{G}_T$ and bilinear map $\e$ are computable in polynomial time with respect to $\lambda$. We refer to $\mathbb{G}$ and $\mathbb{H}$ as the source groups and $\mathbb{G}_T$ as the target group.

\noindent We introduce a new assumption below, which is used to prove the semantic security of the proposed ABET scheme. The new assumption is proven secure in the generic group model \cite{shoup1997lower}. We underline specific elements to show the differences between the new assumption and the original one. 

\begin{definition}[Extended $q$-type Assumption] Given group generators $g \in \mathbb{G}$ and $h \in \mathbb{H}$, define the following distribution:
\begin{eqnarray*}
\adv_{\calA}^{q'} & = &  | \Pr[ \adv( 1^{\lambda}, {\sf par}, D, T_0 ) =1 ] \\
& & - \Pr[ \adv( 1^{\lambda}, {\sf par}, D, T_1 ) =1 ] |, where \\
& & {\sf par} = ( q, \mathbb{G}, \mathbb{H}, \mathbb{G}_T, \e, g, h  ) \leftarrow {\sf GroupGen}(1^{\lambda}) \\
& & a, b, c, d \leftarrow \mathbb{Z}_q^*, s, \{ z \} \leftarrow \Z ; D = ( g^{a}, \underline{h^{b}}, g^{c}, \\
& & g^{(ac)^2}, \underline{ g^{abd}, g^{d/ab}, h^{abd}, h^{abcd}, h^{d/ab}, h^c, h^{cd/ab} }, \\
& & g^{z_i}, g^{acz_i}, g^{ac/z_i}, g^{a^{2}cz_i}, g^{b/z_i^2}, g^{b^{2}/z_i^2},  \forall i \in [q] \\
& & g^{acz_i/z_j}, g^{bz_i/z_j^2}, g^{abcz_i/z_j}, g^{(ac)^2z_i/z_j}, \forall i,j \in [q], \\ 
& & i \neq j ); T_0 = g^{abc} , T_1 = g^s. 
\end{eqnarray*}

\noindent The extended $q$-type (or $q'$-type) assumption is secure if $\adv_{\mathcal{A}} (\lambda)$ is negligible $\lambda$. 

\end{definition}

\noindent The detailed theorem and proof are shown in Appendix A. We also present an Extended Decisional Diffie-Hellman Assumption \cite{tian2020policy}, which is used to prove the ciphertext anonymity of the proposed ABET scheme. 

\begin{definition}[Extended Decisional Diffie-Hellman (eDDH)] Given group generators $g \in \mathbb{G}$ and $h \in \mathbb{H}$, define the following distribution:
\begin{eqnarray*}
\adv_{\calA}^{\text{eDDH}} & = &  | \Pr[ \adv( 1^{\lambda}, {\sf par}, D, T_0 ) =1 ] \\
& & - \Pr[ \adv( 1^{\lambda}, {\sf par}, D, T_1 ) =1 ] |, where \\
& & {\sf par} = ( q, \mathbb{G}, \mathbb{H}, \mathbb{G}_T, \e, g, h  ) \leftarrow {\sf GroupGen}(1^{\lambda}) \\
& & a, b, c \leftarrow \mathbb{Z}_q^*, s \leftarrow \Z ; D = ( g^{a}, g^{b}, g^{ab}, h^c, h^{ab}, \\
& & h^{1/ab}, h^{abc} ); T_0 = h^{c/ab} , T_1 = h^s.
\end{eqnarray*}

\noindent The eDDH assumption is secure if $\adv_{\mathcal{A}} (\lambda)$ is negligible in $\lambda$. 

\end{definition}

\subsection{Attribute-based Encryption}
\label{abe}
\smallskip\noindent\textbf{Access Structure.} Let $\calU$ be an attribute universe. An access structure $\Lambda$ is a collection of non-empty subsets of $\calU$ (i.e., $\Lambda \subseteq 2^{\calU} \backslash \lbrace \phi \rbrace$). It is called monotone if $\forall B, C:$ if $B \in \Lambda$ and $B \subseteq C$ then $C\in \Lambda$.

\label{MSP}
\smallskip\noindent{\bf Monotone Span Program (MSP).} A secret-sharing scheme $\prod$ with domain of secrets realizing access structure is called linear over $\Z$ if: 1) The shares of a secret $s \in \Z$ for each attribute form a vector over $\Z$; 2) For each access structure $\Lambda$ on $\delta$, there exists a matrix ${\bf M}$ with $n_1$ rows and $n_2$ columns called the share-generating matrix for $\prod$. For $\mu = 1,...,n_1$, we define a function $\pi$ labels row $\mu$ of ${\bf M}$ with attribute $\pi(\mu)$ from the attribute universe $\calU$. We consider the column vector $\vec{\nu}=(s,r_2,...,r_{n_2})^{\top}$, where $s \in \Z$ is the secret to be shared and $r_2,...,r_{n_2} \in \Z$ are chosen at random. Then ${\bf M} \vec{\nu} \in \mathbb{Z}^{n_1 \times 1}_q$ is the vector of $n_1$ shares of the secret $s$ according to $\prod$. The share $({\bf M} \vec{\nu})_{\mu}$ belongs to attribute $\pi(\mu)$, where $\mu \in [n_1]$.

According to \cite{beimel1996secure}, every linear secret-sharing scheme has the linear reconstruction property, which is defined as follows: we assume that $\prod$ is a MSP for the access structure $\Lambda$, $\delta' \in \Lambda$ is an authorized set and let $I \subset \{1,2,..., n_1 \}$ be defined as $I=\{ \mu \in [n_1] \wedge \pi(\mu)\in \delta' \}$. There exists the constants $\{\gamma_{\mu} \in \Z \}_{\mu \in I}$ such that for any valid share $\{\lambda_{\mu} = ({\bf M} \vec{\nu})_{\mu} \}_{ \mu \in I}$ of a secret $s$ according to $\prod$, $\sum_{ \mu \in I}\gamma_{\mu} \lambda_{\mu}=s$. Meanwhile, such constants $\{\gamma_{\mu} \}_{\mu \in I}$ can be found in time polynomial in the size of the share-generating matrix ${\bf M}$. For any unauthorized set $\delta''$, no such $\{ \gamma_{\mu} \}$ exist.

\smallskip\noindent\textbf{Attribute-Based Encryption with Public Traceability.} It consists of the following algorithms. We assume an index space $\{ 1, \cdots, k \}$, where $k$ denotes the maximal number of the index. This definition is inspired by \cite{liu2013blackbox}. 

\begin{itemize}
\item {{\sf Setup}$(1^{\lambda})$:} It takes a security parameter $\lambda$ as input, outputs a master key pair $(\msk, \mpk)$.

\item {{\sf KeyGen}$(\msk, \Lambda)$:} It takes the master secret key $\msk$, an access policy as input, outputs a decryption key $ssk_{\Lambda_i}$, which is assigned by a unique index $i$.  
 
\item {{\sf Enc}$(\mpk, m, \delta, j )$:} It takes the master public key $\mpk$, a message $m$, a set of attributes $\delta \in \calU$, and an index $j \in \{ 1, k+1 \}$ as input, outputs a ciphertext $C$. Note that $C$ contains $\delta$, not index $j$.

\item {{\sf Dec}$(\mpk, C, ssk_{\Lambda_i})$:} It takes the master public key $\mpk$, a ciphertext $C$, and the decryption key $ssk_{\Lambda_i}$ as input, outputs the message $m$ if $1 = \Lambda_i (\delta) \wedge j \leq i $. 

\item{{\sf Trace}$(\mpk, \calD, \epsilon )$:} It takes master public key pair $\mpk$, a policy-specific decryption device/blackbox $\calD$, and a parameter $\epsilon > 0$ as input, outputs a set of indexes $\mathbb{K}_T \in \{ 1, \cdots, k \}$, where $\epsilon$ is polynomially related to $\lambda$, and $\mathbb{K}_T$ denotes the index set of the accused decryption keys. 

\end{itemize}

\smallskip\noindent\textbf{Public Traceability.} Given a policy-specific decryption device/blackbox that includes a set of decryption keys, the tracing algorithm, which treats the decryption blackbox as an oracle, can identify the accused decryption keys that have been used in constructing the decryption blackbox. The decryption blackbox is associated with a specific access policy $\Lambda_{\calD}$. Informally speaking, any public user can generate a ciphertext on a message under a set of attributes that satisfies $\Lambda_{\calD}$, and an index $j \in \{ 1, \cdots, k+1 \} $. Then, the public sends the ciphertext to the decryption blackbox and checks whether the decryption is successful. If decryption succeeds, the public outputs the accused index $j$; otherwise, the public generates a new ciphertext under another attribute set and index. The public continues this process until finding a set of accused indexes $\mathbb{K}_T \in \{ 1, \cdots, k \} $. The tracing algorithm's formal definition is referred to \cite{boneh2006fully,boneh2006ful}, which is analogous to the traitor tracing algorithm used in the broadcast encryption.
 
The ABET scheme requires that the encryptor generates a ciphertext on a message associated with a set of attributes and a hidden index $j$. The decryptor can decrypt the message if the set of attributes is satisfied by her access policy, and $j \leq i$. We stress that the hidden index (or index-hiding) is critical to ABET. On the one hand, the index-hiding ensures that a ciphertext generated by the encryptor using an index $j$ reveals no information about $j$. On the other hand, the public user can pick a possible accused index $j \in  \{ 1, \cdots, k+1 \} $ in generating ciphertext for public tracing. In this work, we call index-hiding as ciphertext anonymity. We denote policy-specific decryption blackbox as access blackbox because it accumulates various rewriting privileges for blockchain rewriting.

\subsection{Digital Signature}
\label{homo_sig}

A digital signature scheme $\Sigma $ = ({\sf Setup}, {\sf KeyGen}, {\sf Sign}, {\sf Verify}) is homomorphic, if the following conditions are held. 

\begin{itemize}

	\item{\textit{Simple Key Generation.}} It requires $ (\sk, \pk) \leftarrow {\sf KeyGen} (pp)$ and $pp \leftarrow {\sf Setup}(1^{\lambda})$, where $\pk$ is derived from $\sk$ via a deterministic algorithm $ \pk \leftarrow {\sf KeyGen}'(pp, \sk)$. 
	
	\item{\textit{Linearity of Keys.}} It requires $ {\sf KeyGen}' (pp, \sk+ \Delta(\sk)) = M_{\pk} ( pp, \break {\sf KeyGen}' (pp, \sk), \Delta(\sk)) $, where $M_{\pk}$ denotes a deterministic algorithm which takes $pp$, a verification key $\pk$ and a ``shifted$"$ value $\Delta(\sk)$ as input, outputs a ``shifted$"$ verification key $\pk'$. $\Delta$ denotes the difference or shift between two keys. 
	
	\item{\textit{Linearity of Signatures.}} Two distributions are identical: $\{ \sigma' \leftarrow {\sf Sign}(pp, \sk + \Delta(\sk), m) \} $ and $ \{ \sigma' \leftarrow M_{\Sigma} (pp, \pk , m, \sigma, \Delta(\sk)) \} $, where $\sigma \leftarrow {\sf Sign}(pp, \sk , m) $, and $M_{\Sigma}$ denotes a deterministic algorithm which takes $pp$, a verification key $\pk$, a message-signature pair $(m, \sigma)$ and a ``shifted$"$ value $\Delta(\sk)$ as input, outputs a ``shifted$"$ signature $\sigma'$.
	
	\item{\textit{Linearity of Verifications.}} It requires ${\sf Verify} (pp, M_{\pk}$ ($pp, \pk, \Delta(\sk)), \break m, M_{\Sigma} ( pp, \pk , m, \sigma, \Delta(\sk) ) ) = 1 $, and ${\sf Verify}( pp, \pk, m, \sigma ) = 1$.
	 
\end{itemize}

\noindent The Schnorr signature scheme \cite{schnorr1991efficient} satisfies the homomorphic properties regarding keys and signatures. We rely on this homomorphic property to find the connection between a transaction and its modified versions. 

\subsection{Dynamic Proactive Secret Sharing}
\label{DPSS}
A dynamic proactive secret sharing {\sf DPSS} consists of {\sf Share}, {\sf Redistribute}, and {\sf Open} \cite{baron2015communication} protocols. It allows a dealer to share a secret $s$ among a group of $n_0$ users such that the secret is secure against a {\it mobile} adversary, and allow any group of $n_0$-$t$ users to recover the secret, where $t$ denotes a threshold. The proactive security means that the execution of the protocol is divided into phases (or epochs) \cite{ostrovsky1991withstand}, and a mobile adversary is allowed to corrupt users across all epochs, under the condition that no more than a threshold number of users are corrupted in any given epoch. The {\sf Redistribute} protocol prevents the mobile adversary from disclosing or destroying the secret and allows the set of the users and the threshold to change. Assuming that for each epoch $i$, no more than $t$ users are corrupted during epoch $i$, the following three properties hold:

\begin{itemize}
	\item{\it Termination:} All honest users engaged in the protocol complete each execution of {\sf Share}, {\sf Redistribute}, and {\sf Open}.
	
	\item{\it Correctness:} All honest users output a secret $s'$ upon completing of {\sf Open}, such that $s'$ = $s$ if the dealer was honest during the execution of {\sf Share}. 
	
	\item{\it Secrecy:} If the dealer is honest, then $s$ leaks no information to the adversary. 
	
\end{itemize}

\noindent The definition described in \cite{baron2015communication} is for information-theoretically (or perfectly) secure protocols. We merely require {\sf DPSS} to be computationally secure due to the instantiation used in this work has computational security. Dynamic allows the set of users in a group (or committee) to be dynamically changed, which is useful in the permissionless blockchains. The {\sf Redistribute} protocol has two processes: resharing the key shares to change the committee membership and threshold, and updating the key shares across epochs to tackle mobile adversary. 

\begin{itemize}
	
	\item{\it Resharing the Key Shares \cite{desmedt1997redistributing}.} We rely on a bivariate polynomial to share a secret $s$: $f(x, y) = \underline{s} + a_{0,1}x + a_{1,0}y + a_{1,1}xy + \cdots + a_{t_x, t_y}x^{t_x}y^{t_y} $, where $t_x, t_y$ denote different thresholds. So there are two ways to share the secret $s$:
	
	\begin{enumerate}
	
		\item If we fix $y=0$, then the key shares include $\{ f(i_0, 0), f(i_1, 0), \cdots, \break f(i_{t_x}, 0) \}$;
		
		\item If we fix $x=0$, then the key shares include $\{ f(0, j_0), f(0, j_1), \cdots, \break f(0, j_{t_y}) \}$. 
		
	\end{enumerate}
	
We show how to transfer the ownership of the shareholders from committee $A$ to committee $B$. First, we distribute key shares $\{ f(i,y) \}$ to all users in committee $A$. Second, each user in committee $A$ generates a set of temporary shares by running a secret sharing scheme (SSS) (e.g., Shamir's \cite{shamir1979share}) on his own key share. In other words, his key share is the secret for SSS. Third, users in committee $A$ send those temporary shares to users in committee $B$. Now, users in the committee $B$ accumulate the received temporary shares and obtain another form of key shares $\{ f(x,j) \}$ via interpolation of $t_y$ temporary shares. To this end, the transfer between the two committees is successful. Note that either key shares $\{ f(i,y) \}$ or $\{ f(x,j) \}$ can be used to recover the secret $s$. 
		
		\item{\it Updating the Key Shares \cite{herzberg1995proactive}.} Suppose that a bivariate polynomial is used to share the secret $s$: $f(x, y) = \underline{s} + a_{0,1}x + a_{1,0}y + a_{1,1}xy + a_{0,2}x^2 + a_{2,0}y^2 + a_{2,2}x^2y^2 + \cdots + a_{t_x, t_y}x^{t_x}y^{t_y} $. To update $f(x, y)$, we need another bivariate polynomial: $f'(x, y) = \underline{0} + a_{0,1}'x + a_{1,0}'y + a_{1,1}'xy + \cdots + a_{t_x, t_y}'x^{t_x}y^{t_y} $, which takes 0 as the secret. The reason is that the secret $s$ in $f(x,y)$ will not be changed after updating by $f'(x,y)$. A crucial point is, we allow users in a new committee to collaboratively generate a polynomial $f'(x, y)$, thus the shareholders between old and new committees become independent (to ensure proactive security). Note that $t_x$ may not equal to $t_y$ because the threshold between committees can be different, and we call it asymmetric bivariate polynomial. 
		
\end{itemize}

\subsection{Polynomial Commitments}

A simplified version of polynomial commitment scheme \cite{kate2010constant} is shown as follows.

\begin{itemize}
	
	\item{${\sf Setup}(1^{\lambda}, t)$:} It takes a security parameter $\lambda$ and $t$ as input, outputs a key pair $(\msk, \mpk)$, where $\msk = \alpha$, $\mpk = (g, g^{\alpha}, \cdots, g^{\alpha^t}, \break h, h^{\alpha}, \e)$. 
	
	\item{${\sf Commit}(\mpk, f(x) )$:} It takes the public key $\mpk$, and a polynomial $f(x) = a_0 + a_{1}x + a_{2}x^2 + \cdots + a_{t}x^{t}$ as input, outputs $C = \prod_{j=0}^{t} (g^{\alpha^j})^{a_j}$ as the commitment to $f(x)$. 
	
	\item{${\sf CreateWitness}(\mpk, C, f(x))$:} It takes the public key $\mpk$, and the polynomial $f(x)$ as input, outputs a tuple $(i, f(i), w_i)$. Specifically, it computes a polynomial $\frac{f(x)-f(i)}{x-i}$ (note that the coefficients of the resulting quotient polynomial are $(\widehat{a_0}, \widehat{a_1}, \cdots, \widehat{a}_{t})$), and a witness $w_i = \prod_{j=0}^{t} (g^{\alpha^j})^{ \widehat{a}_j} $.
	
	\item{${\sf VerifyEval}(\mpk, C, i, f(i), w_i)$:} It takes the public key $\mpk$, a commitment $C$, and the tuple $(i, f(i), w_i)$ as input, outputs 1 if $ \e(C/g^{f(i)}, \break h) = \e( w_i, h^{\alpha}/h^i)$. 
	
\end{itemize}

\noindent The witness $w_i$ proves that $f(i)$ is a correct evaluation at $i \in \Z $, without revealing the polynomial $f(x)$. The binding property is based on the $t$-Strong Diffie-Hellman assumption \cite{kate2010constant}, while the hiding property is based on the Discrete Logarithm (DL) assumption. If the KZG commitment scheme is used in DPSS \cite{maram2019churp}, we can hold users accountable in a committee. In particular, the KZG commitment scheme is publicly verifiable if we append users' commitments and witnesses to blockchain. We assume that they are confirmed in the blockchain using Proof-of-Work (PoW) consensus (it is not difficult to extend this assumption to other consensus like Proof of Stake \cite{proof_of_stake}, Proof of Space \cite{dziembowski2015proofs}). In this work, a polynomial commitment scheme is used in DPSS. 

\section{Models and Definitions}

\subsection{System Model}

The system model involves three types of entities: user, modifier, and miner, in which the entities can intersect, such as a user can be a modifier and/or a miner. The communication model considers both on-chain and off-chain settings. The on-chain setting is the permissionless blockchain, where {\it read} is public, but {\it write} is granted to anyone who can show PoW. The off-chain setting assumes that every user has a point-to-point (P2P) channel with every other users. One may use Tor or transaction ghosting to establish a P2P channel \cite{maram2019churp}, and further detail is given in Appendix B. Such P2P channel works in a synchronous model, i.e., any message sent via this channel is received within a known bounded time-period.

\begin{figure}	
\centering
\includegraphics[scale=0.22]{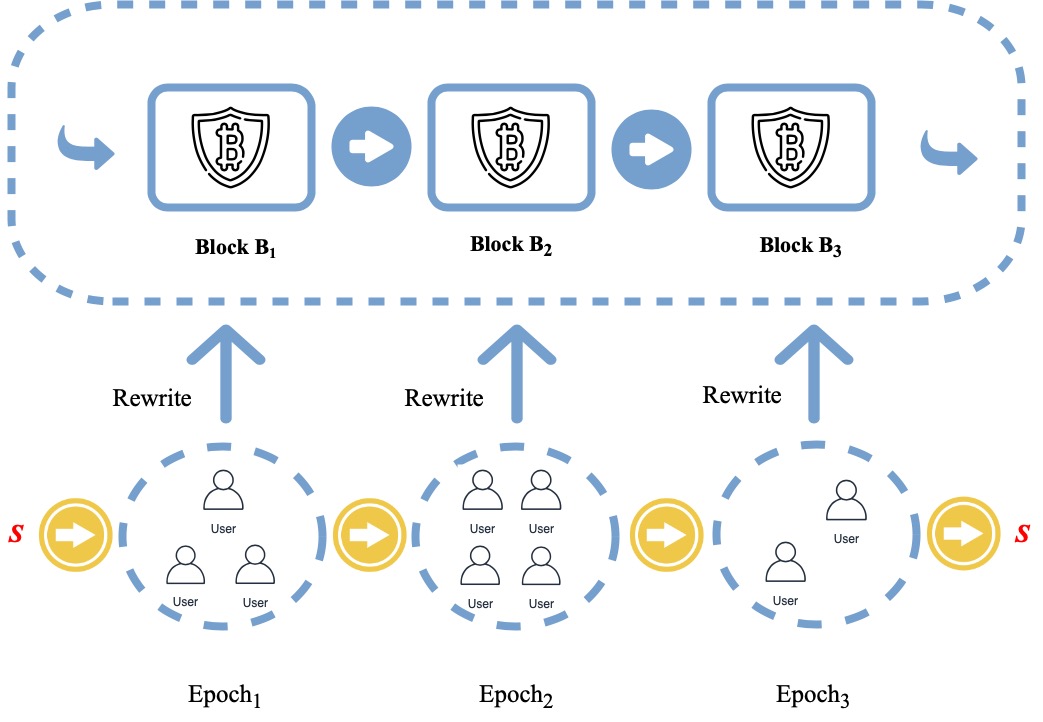}
\centering\caption{Blockchain rewriting with dynamic committees. Users may join in or leave from a committee, and a designated modifier in a committee can rewrite the blockchain. The secret $s$ remains fixed across different committees.}
\label{sys}
\end{figure}

The system proceeds in fixed time periods called epochs. In the first epoch, a committee election protocol (e.g., Algorand's $BA*$ protocol \cite{gilad2017algorand}, or other methods \cite{luu2016secure,zamani2018rapidchain}) is executed, so that a set of users can agree on an initial committee with Byzantine fault tolerance (e.g., up to 1/3 malicious members). The secret $s$ in the initial committee can be generated by an honest user (e.g., committee leader) or in a distributed fashion \cite{gennaro1999secure}. The secret $s$ is shared among the committee members. Similarly, the setup of the commitment scheme can be performed by an honest user in the initial committee.

In Figure \ref{sys}, a blockchain is generated by users who append their hashed contents to the blockchain. Later, modifiers with enough rewriting privileges are required to rewrite the hashed contents. We stress that the link of hash-chain remains intact after rewriting, and the secret remains fixed across different committees. We assume at most $n$ users (i.e., protocol participants) exist in each epoch. We consider $k$ dynamic committees, each of which has a varying number of committee members, and we denote $n_0 \leq n$ as a committee's size. The parameters $n$ and $k$ are independent. We also consider dynamic churn (i.e., join/leave) of the protocol participants. In particular, we do not assume that $k$ committees exist in each different epoch (or we allow several committees to exist in the same epoch).

\noindent\textbf{Remark.} To prevent a malicious user from controlling a committee by launching Sybil attacks \cite{douceur2002sybil}, we rely on the PoW-based identity generation mechanism \cite{luu2016secure,andrychowicz2015pow}. The mechanism allows all users to establish their identities in a committee, yet limiting the number of Sybil identities created by a malicious user. In Elastico \cite{luu2016secure}, each user locally generates/establishes an identity consisting of a public key, an IP address, and a PoW solution. The user must solve a PoW puzzle which has publicly verifiable solutions to generate the final component of the identity. A PoW solution also allows other committee members to verify and accept the identity of a user. Because solving PoW requires computation, the number of identities that the malicious user can create is limited by a fraction of malicious computational power. One can refer to \cite{luu2016secure,zamani2018rapidchain,yu2020ohie} for the detailed discussion on Byzantine fault resiliency. 

\subsection{Definition} 

An accountable fine-grained chameleon hash with dynamic committees consists of the following algorithms.

\begin{itemize}
	\item{${\sf Setup}(1^{\lambda})$:} It takes a security parameter $\lambda$ as input, outputs a master key pair $(\msk, \mpk )$. Note that $\msk$ is shared by committee members.  
	
	\item{${\sf KeyGen}({\sf C}_i^e, \Lambda )$:} It takes a committee ${\sf C}_i^e$, and a policy $\Lambda$ as input, outputs a secret key $\sk_{\Lambda_i}$. The committee index $i \in \{ 1, \cdots, k \}$, where $k$ denotes the total number of committees. 
		
	\item{${\sf Hash}(\mpk, m, \delta, j)$:} It takes the master public key $\mpk$, a message $m \in \calM $, a set of attributes $\delta \in \calU$, and an index $j \in \{ 1, \cdots, k+1 \}$ as input, outputs a chameleon hash $h$, a randomness $r$, and a signature $\sigma$. Note that $\calM = \{ 0, 1 \}^*$ denotes a general message space.  
	
	\item{${\sf Verify}( \mpk, h, m, r, \sigma)$:} It takes the master public key $\mpk$, chameleon hash $h$, message $m$, randomness $r$, signature $\sigma$ as input, output a bit $b \in \{ 0, 1 \}$. 
	 
	\item{${\sf Adapt}(\sk_{\Lambda_i}, h, m, m', r , \sigma )$:} It takes the secret key $\sk_{\Lambda_i}$, chameleon hash $h$, messages $m$ and $m'$, randomness $r$, and signature $\sigma$ as input, outputs $r'$ and $\sigma'$ if $1 = \Lambda (\delta)$ and $ i \leq j$.  
	
	\item{${\sf Judge}( \mpk, T' )$:} It takes the master public key $\mpk$, a modified transaction $T'$ as input, outputs a linked transaction-committee pair $(T', {\sf C}_i^e)$, where $T' = (h, m', r', \sigma' )$. It means a user with a rewriting privilege from committee ${\sf C}_i^e$ has modified transaction $T = (h, m, r, \sigma)$. 
	
\end{itemize}

\noindent\textbf{Correctness.} The definition is {\it correct} if for all security parameters $\lambda$, for all $\delta \in \calU$, all keys $(\msk, \mpk) \leftarrow {\sf Setup} (1^{\lambda})$, for all $\delta \in \Lambda$, for all $i \leq j$, for all $\sk_{\Lambda_i} \leftarrow {\sf KeyGen} ({\sf C}_i^e, \Lambda )$, for all $m \in \calM$, for all $ (h, r, \sigma) \leftarrow {\sf Hash} (\mpk, m , \delta, j ) $, for all $m' \in \calM$, for all $ r' \leftarrow {\sf Adapt} (\sk_{\Lambda_i}, m, m', h, r, \sigma) $, we have $ 1 =  {\sf Verify} (\mpk, h, m, r, \sigma) = {\sf Verify} (\mpk, h, m', r', \sigma' ) $.

\subsection{Security Model}

We consider three security guarantees, including indistinguishability, adaptive collision-resistance, and accountability. 

\begin{itemize}
    \item{\it Indistinguishability.} Informally, an adversary cannot decide whether for a chameleon hash its randomness was freshly generated using {\sf Hash} algorithm or was created using {\sf Adapt} algorithm even if the secret key is known. We define a formal experiment between an adversary $\calA$ and a simulator $\calS$ in Figure \ref{IND}. The security experiment allows $\calA$ to access a left-or-right {\sf HashOrAdapt} oracle, which ensures that the randomness does not reveal whether it was obtained from {\sf Hash} or {\sf Adapt} algorithm. The hashed messages are adaptively chosen from the same message space $\calM$ by $\calA$. 

\begin{figure}
\centering
\scriptsize
	\fbox{\parbox{0.47\textwidth}{
\[\begin{array}{l}
\text{Experiment}~ {\sf Exp}_{ \calA }^{\sf IND}( \lambda ) \\

({\sf C}_i^e, \mpk) \leftarrow {\sf Setup}(1^{\lambda}) , b \leftarrow \{ 0,1 \} \\

b^* \leftarrow \calA^{ {\sf HashOrAdapt}({\sf C}_i^e, \cdots, b) } (\msk ) \\

\quad \text{where}~ {\sf HashOrAdapt}({\sf C}_i^e, \cdots, b)~\text{on input}~ m, m',\delta, \Lambda, j, j' : \\

\quad\quad \sk_{\Lambda_i} \leftarrow {\sf KeyGen}({\sf C}_i^e, \Lambda ) \\

\quad\quad (h_0, r_0, \sigma_0) \leftarrow {\sf Hash}( \mpk, m', \delta, j' ) \\

\quad\quad (h_1, r_1, \sigma_1) \leftarrow {\sf Hash}( \mpk, m, \delta, j ) \\

\quad\quad r_1 \leftarrow {\sf Adapt}(\sk_{\Lambda_i}, m, m', h_1, r_1, \sigma_1 ) \\

\quad\quad \text{return}~ \bot ~\text{if}~ r_0 = \bot \vee r_1 = \bot  \\

\quad\quad \text{return}~ (h_b, r_b, \sigma_b) \\

\text{return 1, if}~ b^* = b; \text{else, return}~ 0.

\end{array}\]
}}
\caption{Indistinguishability.}
\label{IND}
\end{figure}

\noindent We require $1 = \Lambda(\delta) $ and ${\sf Verify} (\mpk, m', h_0, r_0, \sigma_0) = {\sf Verify} (\mpk, m,\break h_1, r_1, \sigma_1) = 1$. Note that $\msk$ is shared by committee members, such that ${\sf C}_i^e = \{ s_i^e \}^{n_0}$. We define the advantage of the adversary as
\begin{equation*}
 \adv_{\mathcal{A} }^{\sf IND}(\lambda ) = |
 \prob[{\sf Exp}_{ \calA }^{\sf IND}(1^{\lambda } ) \rightarrow 1] - 1/2|.
\end{equation*}

\begin{definition} 
The proposed generic framework is indistinguishable if for any probabilistic polynomial-time (PPT) $\calA$, $\adv_{\mathcal{A}}^{\sf IND} (\lambda)$ is negligible in $\lambda$. 
\end{definition}

\item{\it Adaptive Collision-Resistance.} Informally, a mobile adversary can find collisions for a chameleon hash if she possesses a secret key satisfies an attribute set associated with the chameleon hash (this condition is modelled by {\sf KeyGen'} oracle). We define a formal experiment in Figure \ref{Coll}. We allow $\calA$ to see collisions for arbitrary access policies (i.e., {\sf KeyGen''} and {\sf Adapt'} oracles). We also allow $\calA$ to corrupt a threshold number of shareholders (i.e., {\sf Corrupt} oracle) in a committee. Note that the key shares can be transferred between committees while $\msk$ is fixed. 

\begin{figure}
\centering
\scriptsize	
\fbox{\parbox{0.47\textwidth}{
\[\begin{array}{l}
\text{Experiment}~ {\sf Exp}_{ \calA }^{\sf ACR}( \lambda ) \\

({\sf C}_i^e , \mpk) \leftarrow {\sf Setup}(1^{\lambda}), \calQ_1, \calQ_2, \calQ_3 \leftarrow \emptyset \\

( m^*, r^*, m^{*'}, r^{*'}, h^*, \sigma^*, \sigma^{*'} ) \leftarrow \calA^{ \calO } ( \mpk ) \\

\quad \text{where}~ \calO \leftarrow \{ {\sf KeyGen'}, {\sf KeyGen''}, {\sf Hash'}, {\sf Adapt'}, {\sf Corrupt} \} \\

\quad \text{and}~{\sf KeyGen'}({\sf C}_i^e, \cdot )~\text{on input}~ \Lambda: \\
 
\quad\quad \msk_{\Lambda_i} \leftarrow {\sf KeyGen}({\sf C}_i^e, \Lambda) \\

\quad\quad \calQ_1 \leftarrow \calQ_1 \cup \{ \Lambda \} \\

\quad\quad \text{return}~\msk_{\Lambda_i} \\

\quad \text{and}~{\sf KeyGen''}({\sf C}_i^e, \cdot )~\text{on input}~ \Lambda: \\
 
\quad\quad \msk_{\Lambda_i} \leftarrow {\sf KeyGen}({\sf C}_i^e, \Lambda) \\

\quad\quad \calQ_2 \cup \{ (i, \msk_{\Lambda_i} ) \} \\

\quad\quad i \leftarrow i+1 \\
 
\quad \text{and}~{\sf Hash'}(\mpk, \cdots )~\text{on input}~ m, \delta, j: \\
 
\quad\quad (h, r, \sigma) \leftarrow {\sf Hash}( \mpk, m, \delta, j ) \\

\quad\quad \calQ_3 \leftarrow \calQ_3 \cup \{ (h, m, \delta ) \} \\

\quad\quad \text{return}~(h, r, \sigma) \\

\quad \text{and}~{\sf Adapt'}(\mpk, \cdots )~\text{on input}~ m, m', h, r, i, \sigma : \\

\quad\quad \text{return}~\bot, \text{if}~ (i, \msk_{\Lambda_i}) \notin \calQ_2 ~\text{for some}~ \msk_{\Lambda_i} \\

\quad\quad r' \leftarrow {\sf Adapt} ( \msk_{\Lambda_i}, m, m', h, r, \sigma ) \\

\quad\quad \text{if}~(h, m, \delta) \in \calQ_3 ~\text{for some}~\delta, \\

\quad\quad \text{let}~ \calQ_3 \leftarrow \calQ_3 \cup \{ (h, m', \delta) \} \\

\quad\quad \text{return}~r' \\

\quad \text{and}~ {\sf Corrupt} (\mpk, \cdot )~\text{on input}~ {\sf C}_i^e : \\

\quad\quad \text{return}~ \{ s_i^e \}^t  \\

\text{return 1, if} \\

1 = {\sf Verify} (\mpk, m^*, h^*, r^*, \sigma^*) = {\sf Verify} (\mpk, m^{*'}, h^*, r^{*'}, \sigma^{*'}) \\

\wedge (h^*, \cdot, \delta) \in \calQ_3, \text{for some}~\delta \wedge m^* \neq m^{*'} \\

 \wedge \delta \cap \calQ_1 = \emptyset \wedge (h^*, m^*, \cdot ) \notin \calQ_3 \\

\text{else, return}~ 0. \\

\end{array}\]

}}

\caption{Adaptive Collision-Resistance.}
\label{Coll}
\end{figure}

\noindent $\calA$ is not allowed to corrupt more than a threshold number of shareholders in any committee. We define the advantage of the adversary as
\begin{equation*}
\adv_{\mathcal{A} }^{\sf ACR}(\lambda ) = \prob[{\sf Exp}_{ \calA }^{\sf ACR}(1^{\lambda } ) \rightarrow 1].
\end{equation*}

\begin{definition} 
The proposed generic framework is adaptively collision-resistant if for any PPT $\calA$, $\adv_{\mathcal{A}}^{\sf ACR} (\lambda)$ is negligible in $\lambda$. 
\end{definition}

\item{\it Accountability.} Informally, an adversary cannot generate a bogus message-signature pair for a chameleon hash, which links a user to an accused committee, but the user has never participated in the accused committee. We define a formal experiment in Figure \ref{Acc}. We allow $\calA$ to see whether a modified transaction links to a committee (i.e., {\sf Judge'} oracle). Let set $\calQ$ record the transactions produced by the {\sf Judge'} oracle. 
\begin{figure}
\centering
\scriptsize
\fbox{\parbox{0.47\textwidth}{
\[\begin{array}{ll}
\text{Experiment}~ {\sf Exp}_{ \calA }^{\sf ACT}( \lambda  ) \\

({\sf C}_i^e, \mpk) \leftarrow {\sf Setup}(1^{\lambda}), \mathcal{Q} \leftarrow \emptyset \\

%T \leftarrow {\sf Hash}( \mpk, m, \delta, j ) \\

T^{*} \leftarrow \calA^{ {\sf Judge'}( {\sf C}_i^e , \cdots ) } (\mpk) \\

\quad \text{where}~ {\sf Judge'}({\sf C}_i^e, \cdots )~\text{on input}~ T, \Lambda, m' : \\

\quad\quad \msk_{\Lambda_i} \leftarrow {\sf KeyGen}({\sf C}_i^e, \Lambda) \\

\quad\quad r' \leftarrow {\sf Adapt}(\msk_{\Lambda_i}, T, m' ) \\

\quad\quad \calQ \leftarrow \calQ \cup \{ (T, T' ) \} \\

\quad\quad (T' , {\sf C}_i^e ) \leftarrow {\sf Judge} ({\sf C}_i^e, T')  \\

\quad\quad \text{return}~ (T' , {\sf C}_i^e ) \\

\text{return}~ 1, \text{if}~( T^{*}, {\sf C}^*) \wedge (T^* \notin \calQ \vee \pk^* \notin {\sf C}^* )  \\

\text{else, return 0}.

\end{array}\]
}}

\caption{Accountability.}	
\label{Acc}
\end{figure}

\noindent We denote $T = (h, m, r, \sigma)$ and $T' = (h, m', r', \sigma')$ as original and modified transactions with respect to chameleon hash $h$. We also denote the linked transaction-committee pair as $( T', {\sf C}_i^e )$. We define the advantage of the adversary as
\begin{equation*}
\adv_{\mathcal{A} }^{\sf ACT}(\lambda ) = \prob[{\sf Exp}_{ \calA }^{\sf ACT}(1^{\lambda } ) \rightarrow 1].
\end{equation*}

\begin{definition} 
The proposed generic framework is accountable if for any PPT $\calA$, $\adv_{\mathcal{A}}^{\sf ACT} (\lambda)$ is negligible in $\lambda$. 
\end{definition}

\end{itemize}

\section{Generic Construction}

The proposed generic construction consists of the following building blocks.

\begin{itemize}
	\item A chameleon hash scheme {\sf CH} = {\sf (Setup, KeyGen, Hash, Verify, Adapt)}.
	
	\item An attribute-based encryption scheme with public traceability {\sf ABET} = {\sf (Setup, KeyGen, Enc, Dec, Trace)}. 
	
	\item A dynamic proactive secret sharing scheme {\sf DPSS} = ({\sf Share}, {\sf Redistribute}, {\sf Open}). 
	
	\item A digital signature scheme $\Sigma$ = ({\sf Setup}, {\sf KeyGen}, {\sf Sign}, {\sf Verify}).
	
\end{itemize}

\noindent\textbf{High-level Description.} We assume that every user has a key pair $(\sk, \pk)$ and that users' public keys (i.e., users' identities) are known to all users in a committee. Meanwhile, each user possesses a set of attributes. In particular, more than a threshold number of users in a committee can collectively grant a rewriting privilege to a modifier based on her attribute set. A user with $\pk$ creates a transaction $T$ that includes a chameleon hash, a ciphertext under his attribute set, and a signature (i.e., signs $T$ using his secret key $\sk$). A modifier with $\pk'$, who is granted the rewriting privilege from a committee, is allowed to rewrite the transaction $T$, and signs the modified transaction using her secret key $\sk'$. We assume the $t$-out-of-$n_0$ {\sf DPSS} scheme to be executed over off-chain P2P channels and let all $k$ committees have the same parameters ($t, n_0$). The proposed construction is shown below.

\begin{itemize}
	\item{${\sf Setup}(1^{\lambda})$:} A user takes a security parameter $\lambda$ as input, outputs a public parameter $ {\sf PP} = (\mpk_{\sf ABET}, {\sf PP}_{\Sigma}, {\sf PP}_{\sf CH})$, and a secret key $\msk_{\sf ABET}$, where $(\msk_{\sf ABET}, \mpk_{\sf ABET}) \leftarrow {\sf Setup}_{\sf ABET}(1^{\lambda})$, ${\sf PP}_{\Sigma} \leftarrow {\sf Setup}_{\Sigma}( 1^{\lambda} )$, $ {\sf PP}_{\sf CH} \leftarrow {\sf Setup}_{\sf CH} (1^{\lambda}) $. The key shares $\{ s_0 \}^{n_0} \leftarrow {\sf Share}_{\sf DPSS}(\msk_{\sf ABET}) $ are distributed to users within committee ${\sf C}^{0}$, where each user holds a key share, and a key pair $(\sk, \pk) \leftarrow {\sf KeyGen}_{\Sigma}( {\sf PP}_{\Sigma}) $.
	
    \item{${\sf KeyGen}( {\sf C}_i^e, \Lambda )$:} A group of $t$+1 users in committee ${\sf C}_i^e$ take their secret shares $\{ s_e \}^{t+1}$ and a policy $\Lambda$ as input, output a secret key $\sk_{\Lambda_i}$ for a modifier, where $\sk_{\Lambda_i} \leftarrow {\sf KeyGen}_{\sf ABET}(\msk_{\sf ABET}, \Lambda )$, $\msk_{\sf ABET} \leftarrow {\sf Open}_{\sf DPSS} ( \{ s_e \}^{t+1} ) $, and secret shares $\{ s_e \}^{n_0} \leftarrow {\sf Redistribute}_{\sf DPSS} ( \{ s_{e-1} \}^{n_0} )$. Note that one of $t$+1 users can be the modifier. 
		
	\item{${\sf Hash}({\sf PP}, m, \delta, j)$:} A user appends a message $m$, a set of attributes $\delta$, and an index $j$ to the blockchain, performs the following operations
	
	\begin{enumerate}
		
		\item generate a chameleon hash $ (h_{\sf CH}, {\sf r}) \leftarrow {\sf Hash}_{\sf CH}( \pk_{\sf CH} , m) $, where $ (\sk_{\sf CH}, \pk_{\sf CH}) \leftarrow {\sf KeyGen}_{\sf CH} ( {\sf PP}_{\sf CH} )$.
				
		\item generate a ciphertext $C \leftarrow {\sf Enc}_{\sf ABET}(\mpk_{\sf ABET}, \sk_{\sf CH}, \break \delta, j )$, where $\sk_{\sf CH}$ denotes the encrypted message.  
		
		\item generate a message-signature pair $(c, \sigma_{\Sigma})$, where $\sigma_{\Sigma} \leftarrow {\sf Sign}_{\Sigma} \break (\sk, c) $, and message $c$ is derived from $\sk$ and $\sk_{\sf CH}$.
	
		\item output $( h, m, {\sf r}, \sigma)$, where $h \leftarrow (h_{\sf CH}, \pk_{\sf CH}, C )$, and $\sigma \leftarrow (c, \sigma_{\Sigma})$.
	\end{enumerate}
	
	\item{${\sf Verify}({\sf PP}, h, m, {\sf r}, \sigma)$:} It outputs 1 if $1 \leftarrow {\sf Verify}_{\sf CH}( \pk_{\sf CH}, m, h_{\sf CH}, {\sf r} )$ and $1 \leftarrow {\sf Verify}_{\Sigma}(\pk, c, \sigma_{\Sigma})$, and 0 otherwise.
	
	\item{${\sf Adapt}(\sk_{\Lambda_i}, h, m, m', {\sf r}, \sigma )$:} A modifier with a secret key $ \sk_{\Lambda_i}$ and a new message $m'$, performs the following operations
	
	\begin{enumerate}
		
		\item check $ 1 \iseq {\sf Verify}({\sf PP}, h, m, {\sf r}, \sigma ) $.
		
		\item compute $\sk_{\sf CH} \leftarrow {\sf Dec}_{\sf ABET} (\mpk_{\sf ABET}, C, \sk_{\Lambda_i})$.
        
		\item compute a new randomness ${\sf r'} \leftarrow {\sf Adapt}_{\sf CH}( \sk_{\sf CH}, m, m', \break h, {\sf r})$.
		
		\item generate a ciphertext $ C' \leftarrow {\sf Enc}_{\sf ABET} ( \mpk_{\sf ABET}, \sk_{\sf CH}, \break \delta, j )$.
	
		\item generate a message-signature pair $(c', \sigma_{\Sigma}')$, where $c'$ is derived from $\sk'$ and $\sk_{\sf CH}$.
		
		\item output $(h, m', {\sf r'}, \sigma')$, where $h \leftarrow (h_{\sf CH}, \pk_{\sf CH}, C' )$, and $\sigma' \leftarrow (c', \sigma_{\Sigma}')$.
	\end{enumerate}
	
	\item{${\sf Judge} ( {\sf PP}, T' )$:} It takes the public parameter {\sf PP}, and a modified transaction $T'$ as input, outputs a transaction-committee pair $(T', {\sf C}_i^e)$ if the modified transaction $T'$ links to a committee ${\sf C}_i^e$, where $T' = ( h, m', {\sf r}', \sigma' )$. 
				
\end{itemize}

\noindent\textbf{Correctness.} The ${\sf Judge}$ algorithm allows any public user to identify the responsible modifiers and committees given a modified transaction. The modifier (or modifier's public key) is publicly known because a digital signature is used in the construction. Below, we explain the judge process in detail.

First, any public user verifies a connection between a transaction $T$ and its modified version $T'$. The connection can be established, as both message-signature pair $(c, \sigma_{\Sigma})$ in $T$ and message-signature pair $(c', \sigma_{\Sigma}')$ in $T'$, are derived from the same chameleon trapdoor $\sk_{\sf CH}$. Since different modifiers may modify the same transaction, the chameleon trapdoor $\sk_{\sf CH}$ is used in many modified versions of a transaction. Here, we consider a single modified transaction $T'$ for simplicity.

Second, any public user obtains a set of accused committees from interacting with an access blackbox $\calO$, such that $ \{ {\sf C} \} \leftarrow {\sf Trace}_{\sf ABET} ( \mpk_{\sf ABET}, \calO, \epsilon ) $. Specifically, the public sends a ciphertext encrypted a message under a set of attributes (that satisfies the access privilege involved in $\calO$) and a committee index $j \in \{ 1, \cdots, k+1 \}$ to $\calO$. Then, the public outputs the committee index $j$ (we call it accused committee) if decryption succeeds. The public repeats this tracing procedure until output all accused committees. 

Third, if a user with $\pk'$ acts as a modifier in an accused committee, the public outputs $(T', {\sf C})$. It means that a transaction $T'$ is indeed modified by the user $\pk'$ whose rewriting privilege is granted from committee ${\sf C}$. Because we allow the commitment scheme to be used in {\sf DPSS}, the user $\pk'$ is held accountable in a committee. More specifically, user $\pk'$ joins in committee ${\sf C}$ by showing a commitment on his key share to other committee members, and further detail is given in the instantiation. If user $\pk'$ acts as modifiers for many accused committees, the public outputs $(T', \{ {\sf C} \})$. However, if user $\pk'$ did not join in any accused committees, the public still outputs the indexes of the accused committees. This is the second case of blockchain rewriting: an unauthorized user has no granted rewriting privileges from any committee but holds an access blackbox.

To conclude, we achieve public accountability via three steps: 1) Verify a modified transaction; 2) Find an accused committee; 3) Link the modified transaction to the accused committee. We also consider a committee has multiple modifiers with different rewriting privileges, but they should have the same committee index. In this case, the public still identifies the modifiers (i.e., holding different rewriting privileges) in the same committee as the modifiers sign the modified transactions using their secret keys. 
  
\smallskip\noindent\textbf{Remark.} One may notice that a modifier $\pk'$ can assign a new set of attributes $\delta'$ to a modified transaction $T'$ in the ${\sf Adapt}$ algorithm. In other words, the attribute set associated with a mutable transaction can be dynamically updated to satisfy different security requirements in case blockchain system evolves. Therefore, we remark that the proposed generic construction supports fine-grained and flexible blockchain rewriting. Besides, the modifier can use a new index to create $T'$. One may also notice that such flexibility could be misused. The modifier $\pk'$ may intend to rewrite the transaction $T$ with malicious content and change the rewriting privileges to disallow others to modify the transaction. We argue that more than a threshold number of committee members can collectively reset the transaction $T$'s attribute set if such malicious behavior happens. If flexibility is not desired in accountable blockchain rewriting, one can let the transaction owner sign the embedded attribute set so that the modifier rewrites the transaction $T$ only without changing its rewriting privilege. 

The second remark is that we only allow blockchain rewriting that does not affect a transaction's past and future events. If a modifier removes a transaction entirely or changes spendable data of a transaction, it may lead to serious transaction inconsistencies in the chain \cite{deuber2019redactable}. The last remark is, the committee members can reconstruct the secret $\msk_{\sf ABET}$ as DPSS includes the {\sf Open} algorithm. If any committee member rewrites blockchain maliciously, he/she is held accountable because the modified transaction is signed by him/her.

\smallskip\noindent\textbf{Security Analysis.} We show the security result of our proposed construction, and the detailed proofs are referred to Appendix C. 

\begin{theorem}\label{the_ind}	
	The proposed generic framework is indistinguishable if the {\sf CH} scheme is indistinguishable. 
\end{theorem} 

\begin{theorem}\label{the_cr}	
	The proposed generic framework is adaptively collision-resistant if the {\sf ABET} scheme is semantically secure, the {\sf CH} scheme is collision-resistant, and the {\sf DPSS} scheme has secrecy.
\end{theorem} 

\begin{theorem}\label{the_acc}	
	The proposed generic framework is accountable if the $\Sigma$ scheme is existential unforgeability under chosen message attack (EUF-CMA) secure, and the {\sf DPSS} scheme has correctness. 
\end{theorem} 

\section{Instantiation}
\label{instantiation}

\subsection{The Proposed ABET Scheme}

For constructing a practical ABET, we require that the underlying KP-ABE scheme should have minimal number of components in master secret key, while the size of the ciphertext is constant (i.e., independent of the number of committees). Therefore, we rely on a KP-ABE scheme \cite{rouselakis2013practical} and a hierarchy identity-based encryption (HIBE) scheme \cite{boneh2005hierarchical}. The KP-ABE \cite{rouselakis2013practical} can be viewed as the stepping stone to construct ABET. It has a single component in master secret key $\msk$, which requires a single execution of the DPSS. It works in prime-order group, and its security is based on $q$-type assumption in the standard model. One may use more efficient ABE schemes such as \cite{agrawal2017fame}. However, the master secret key $\msk$ has several components, which requires multiple executions of the DPSS. The HIBE \cite{boneh2005hierarchical} has constant-size ciphertext. The ciphertext has just three group elements, and the decryption requires only two pairing operations. In particular, HIBE has one master secret key, which can be shared with KP-ABE. The ABET has been studied in \cite{liu2013blackbox,ning2016traceable,lai2018making,tian2020policy}. Their schemes are based on cipher-policy ABE. They cannot be applied to open blockchains with decentralized access control (i.e., a threshold number of committee members grant decryption keys to users). Our proposed ABET scheme makes it possible based on the KP-ABE scheme.

The intertwined ABET scheme is not anonymous because its ciphertext reveals user's committee information to the public. We extend the intertwined ABET scheme into an anonymous one using asymmetric pairings, i.e., $\e: \mathbb{G} \times \mathbb{H} \rightarrow \mathbb{G}_T$ (as described in \cite{ducas2010anonymity}). The basic idea is, the index-based elements in a modifier's decryption key belong to group $\mathbb{G}$. The index-based elements in a ciphertext belong to group $\mathbb{H}$ so that the ciphertext can conceal the committee's index if the master secret key is unknown. The concrete construction of ABET is embedded in the instantiation. Below, we present Theorem \ref{theo_ABET} to show the proposed ABET scheme has semantic security and ciphertext anonymity. The security analysis is referred to Appendix D.

\begin{theorem}\label{theo_ABET}
The proposed ABET scheme achieves semantic security and ciphertext anonymity, if the $q'$-type assumption and eDDH assumptions hold in the asymmetric pairing groups. 	
\end{theorem}

\subsection{Instantiation}

First, we use the proposed ABET scheme to construct our instantiation. Specifically, the {\sf Setup} and {\sf KeyGen} algorithms in ABET are directly used in the instantiation. The {\sf Enc} and {\sf Dec} algorithms in ABET are part of {\sf Hash} and {\sf Adapt}, respectively. Second, we rely on a recent work \cite{maram2019churp} to initiate the DPSS scheme. We particularly show an instantiation of DPSS with pessimistic case, which allows users to be held accountable in a committee using KZG commitment \cite{kate2010constant}.

We denote an index space as $\{ I_1, \cdots, I_k \} \in ({\Z})^k$, which is associated with $k$ committees. We define a hierarchy as follows: index $i$ is close to the root node $k$, and index $j$ is close to the leaf node. We assume each committee has $n_0$ users, and the threshold is $t$, where $t < n_0/2$ according to \cite{maram2019churp}.  Let $ \h: \{ 0,1 \}^*  \rightarrow \Z $ be a hash function, and the size of hash output $\h$ is assumed to be $l$. Let $\e: \mathbb{G} \times \mathbb{H} \rightarrow \mathbb{G}_T$ be a bilinear pairing. The concrete instantiation is shown below.

\begin{itemize}
	\item{${\sf Setup}(1^{\lambda})$:} It takes a security parameter $\lambda$ as input, outputs a master public key $\mpk = (g, u, v, w, h, \e(g, h)^{\alpha}, \{ g_1^{\alpha}, \cdots g_k^{\alpha} \}, \{ h_1^{\alpha}, \cdots h_k^{\alpha} \},\break g^{\beta}, h^{1/ \alpha}, h^{\beta/\alpha}, \e(g, h)^{\theta /\alpha } )$, and a master secret key $\msk =(\alpha, \beta, \theta) $, where $(\alpha, \beta, \theta ) \in \mathbb{Z}_q^*$ $\{ z_1, \cdots, z_k \} \in \Z$, $g$ is generator of group $\mathbb{G}$, $h$ is generator of group $\mathbb{H}$, $(u, v, w) \in \mathbb{G}$, $ \{ g_1, \cdots, g_k \} = \{ g^{ z_1}, \cdots, g^{ z_k } \} $, $ \{ h_1, \cdots, h_k \} = \{ h^{z_1 }, \cdots, h^{z_k} \} $. Note that the key shares of $\alpha$ and $\theta$ are distributed to users in committee ${\sf C}^0$, and these key shares can be redistributed between dynamic committees (see correctness below).
	
	\item{${\sf KeyGen}( {\sf C}_i^e , ({\bf M}, \pi) )$:} It inputs a committee ${\sf C}_i^e$ with index $(I_1, \cdots, I_i)$, and an access structure $( {\bf M} , \pi )$ (${\bf M}$ has $n_1$ rows and $n_2$ columns), outputs a secret key $\sk_{\Lambda_i}$ for a modifier. Specifically, a group of $t$+1 users in the committee ${\sf C}_i^e$ first recover secrets $\alpha$ and $\theta$. Then, they pick $\{ t_1, \cdots, t_{n_1} \} \in \Z$, for all $ \tau \in [n_1]$, compute $ \sk_{(\tau, 1)} = g^{ s_{\tau} } w^{t_{\tau}} , \sk_{(\tau, 2)} = (u^{\pi(\tau)} v )^{-t_{\tau}} , \sk_{(\tau, 3)} = h^{t_{\tau}} $, where $s_{\tau}$ is a key share from $\alpha$. Eventually, they pick $\{ r_1, \cdots, r_{n_1} \} \in \Z$, compute $\sk_0= ( g^{t^* /\alpha}, g^{r^*}) , \sk_1 = g^{\theta} \cdot \widehat{i}^{t^*} \cdot g^{ \beta \cdot r^*}, \sk_2 = \{ g_{i-1}^{\alpha \cdot t^*}, \cdots g_1^{\alpha \cdot t^*} \} $, where $t^* = \sum_{\tau \in |n_1|} (t_{\tau}) $, $r^* = \sum_{\tau \in |n_1|} (r_{\tau}) $, and $\widehat{i} = g_k^{ \alpha I_1} \cdots g_i^{\alpha I_i} \cdot g \in \mathbb{G} $ is associated with a committee ${\sf C}_i^e$ with index $(I_1, \cdots, I_i)$. The secret key is $\sk_{\Lambda_i} = ( \{ \sk_{\tau} \}_{\tau \in [n_1]}, \sk_0, \sk_1, \sk_2 )$. 
	
	\item{${\sf Hash}(\mpk, m, \delta, j )$:} To hash a message $m \in \Z$ under a set of attributes $\delta $, and an index $(I_1, \cdots I_j)$, a user performs the following operations
	
		\begin{enumerate}
			
			\item choose a randomness ${\sf r} \in \mathbb{Z}_q^*$, and a trapdoor ${\sf R}$, compute a chameleon hash $b = g^m \cdot p'^{\sf r} $ where $p' = g^e$, $e =  \h({\sf R})$. Note that ${\sf R}$ denotes a short bit-string. 
			
			\item generate a ciphertext on message $M = {\sf R}$ under a set of attributes $\delta = \{ A_1, \cdots, A_{|\delta|} \}$ and index $(I_1, \cdots I_j)$. It first picks $s, r_1, r_2, \cdots, r_{|\delta|} \in \Z$, for $\tau \in {|\delta|}$ computes $ ct_{(\tau, 1)} = h^{r_{\tau}}$ and $ct_{(\tau, 2)} = (u^{A_{\tau}} v)^{r_{\tau}}w^{-s}$. Then, it computes $ct = ( {\sf R} || 0^{ l - |{\sf R}|} ) \oplus \h _2( \e(g, h)^{\alpha s} || \e(g, h)^{ \theta s /\alpha})$, $ct_0 = (h^s, h^{s/\alpha}, h^{ \beta \cdot s /\alpha } )$, and $ct_1 = \widehat{j}^{s}$, where $\widehat{j} = h_k^{\alpha I_1} \cdots h_j^{ \alpha I_j} \cdot h \in \mathbb{H} $. Eventually, it sets $C = (ct, \{ ct_{(\tau, 1)}, ct_{(\tau, 2)} \}_{\tau \in [\delta]}, ct_0, ct_1 )$. 
			
			\item generate a signature $epk = g^{esk}, \sigma = esk + \sk \cdot \h( epk || c ) $, where $(esk, epk)$ denotes an ephemeral key pair, and $c = g^{\sk + ( {\sf R} || 0^{ l - |{\sf R}|} )} $ denotes a signed message.
				
            \item output $(m, p', b, {\sf r}, C, c, epk, \sigma)$.

	\end{enumerate}
	
	\item{${\sf Verify}( \mpk, m, p', b, {\sf r}, c, epk, \sigma )$:} Any public user can verify whether a given hash $(b, p')$ is valid, it outputs 1 if $ b = g^{m} \cdot p'^{\sf r} $, and $ g^{\sigma} = epk \cdot \pk^{\h(epk||c)} $. 
	
	\item{${\sf Adapt}(\sk_{\Lambda_i}, m, m', p', b, {\sf r}, C, c, epk, \sigma)$:} A modifier with a secret key $\sk_{\Lambda_i}$, and a new message $m' \in \Z$, performs the following operations
	
		\begin{enumerate}
		
		\item check $ 1 \iseq {\sf Verify} (\mpk, m, p', b, {\sf r}, c, epk, \sigma ) $.
		
		\item run the following steps to decrypt trapdoor ${\sf R}$:

		\begin{enumerate}
			
		\item generate a delegated key w.r.t an index $(I_1, \cdots I_{i+1}) $. It picks $t' \in \Z$, computes $\sk_0= ( g^{ (t^*+t') /\alpha}, g^{r^*}) , \sk_1 = g^{\theta} \cdot \widehat{i}^{t^*} \cdot g^{ \beta \cdot r^*} \cdot ( g_{i-1}^{\alpha \cdot t^*} )^{I_{i+1}} \cdot ( g_k^{\alpha \cdot I_1} \cdots g_{i-1}^{\alpha \cdot I_{i+1}} \cdot g )^{t'} , \sk_2 = \{ g_{i-2}^{\alpha \cdot t^*} \cdot g_{i-2}^{\alpha \cdot t'}, \cdots g_1^{\alpha \cdot t^*} \cdot g_1^{\alpha \cdot t'} \} $. The delegated secret key is $\sk_{\Lambda_{i+1}} = ( \{ \sk_{\tau} \}_{\tau \in [n_1]}, \sk_0, \sk_1, \sk_2 )$. 
		
		\item if the attribute set $\delta$ involved in the ciphertext satisfies the policy MSP $( {\bf M}, \pi )$, then there exists constants $ \{ \gamma_{\mu} \}_{\mu \in I}$ that satisfy the equation in Section \ref{MSP}. It computes  $B$ as follows.  
\begin{eqnarray*}
B & = & \prod_{\mu \in I} ( \e(\sk_{(\mu,1)} , ct_{(0,1) }) \e(\sk_{(\mu,2)} , ct_{(\mu,1)}) \\
& &  \e( ct_{(\mu,2)} , \sk_{(\mu,3)}) )^{\gamma_{\mu}} \\
& = & \e(g, h)^{s \sum_{\mu \in I} \gamma_{\mu} s_{\mu} } =  \e(g, h)^{ \alpha s }, \\
& & where~ \sum_{\mu \in I} \gamma_{\mu} s_{\mu} = \alpha.
\end{eqnarray*}

		\item check $ ({\sf R} || 0^{l - |{\sf R}| } ) \iseq ct \oplus \h ( B || A )$, where $A = \break \frac{\e( \sk_1, ct_{(0,2)})}{\e( \sk_{(0,1)} , ct_1 ) \e( \sk_{(0,2)} , ct_{(0,3)} ) }$. The format ``$|| 0^{l - |{\sf R}|}$" is used to check when the encrypted value ${\sf R}$ is decrypted successfully with certainty $1-2^{ l - |{\sf R}| } $. If the encrypted value is decrypted, then the delegation procedure terminates (note that $B$ is computed once). Note that $\sk_{(0,1)},\sk_{(0,2)}$ denote the first, and second element of $\sk_0$, and the same rule applies to $ct_0$. 
		
		\end{enumerate}

		\item compute a new randomness ${\sf r'} = {\sf r} + (m - m' )/ e $, where $e = \h ({\sf R})$.
		
        \item generate a new ciphertext $C'$ on the same message $M = {\sf R} $ using the attribute set $\delta$ and index $(I_1, \cdots, I_j)$.

        \item generate a signature $ epk' = g^{esk'}, \sigma' = esk' + \sk' \cdot \h (epk' || c' ) $, where $c' = g^{\sk' + ( {\sf R} || 0^{ l - |{\sf R}|} )} $. 
        
		\item output $(m', p', b, {\sf r'}, C', c', epk', \sigma')$.
	\end{enumerate}	
	
\end{itemize}

\noindent{\bf Correctness of DPSS scheme.} Two secrets need to be distributed: $(\alpha, \theta)$. We specifically show users in committee ${\sf C}^{e-1}$ securely {\it handoff} their key shares of secret $\alpha$ to users in committee ${\sf C}^{e}$. According to the DPSS scheme in \cite{maram2019churp}, an asymmetric bivariate polynomial is used: $f(x, y) = \underline{\alpha} + a_{0,1}x + a_{1,0}y + a_{1,1}xy + a_{1,2}xy^2 + \cdots + a_{t, 2t}x^{t}y^{2t} $. Each user in committee ${\sf C}^{e-1}$ holds a {\it full} key share after running {\sf Share} protocol. For example, a user with $\pk$ holds a key share $f(i,y)$, which is a polynomial with dimension $t$. Overall, the handoff (i.e., {\sf Redistribute} protocol) includes three phases: share reduction, proactivization, and full-share distribution. 

\begin{itemize}
	\item{\it Share Reduction.} It requires each user in committee ${\sf C}^{e-1}$ reshares its full key share. For example, user $\pk$ derives a set of {\it reduced} shares $\{ f(i,j) \}_{j \in [1, n_0]}$ from its key share $f(i,y)$ using SSS. Then, each user distributes the reduced shares to users in committee ${\sf C}^{e}$, which includes a user with $\pk'$. As a result, each user in ${\sf C}^{e}$ obtains a reduced share $f(x, j)$ by interpolating the received shares $\{ f(i,j) \}_{i \in [1, t]}$. Note that the dimension of $f(x,j)$ is $2t$, and $2t$+1 of these reduced key shares $\{ f(x,j) \}_{j \in [1, 2t+1]}$ can recover $\alpha$ (see Section \ref{DPSS}). The goal of this dimension-switching (from $t$ to $2t$) is to achieve optimal communication overhead, such that only $2t$+1 users in committee ${\sf C}^{e}$ are required to update $f(x,j)$. 
	
	\item{\it Proactivization.} It requires $F(x,j) = f(x,j) + f'(x,j)$, where $f'(x, y)$ is a new asymmetric bivariate polynomial with dimension $(t, 2t)$ and $f'(0,0) = 0$. We provide more details of $f'(x,y)$ later. 
	
	\item{\it Full-share Distribution.} It requires each user in committee ${\sf C}^{e}$ to recover its full key share with dimension $t$. For example, a full key share $F(i, y)$ is recovered by interpolating the reduced shares $\{ F(i,j) \}_{j \in [1, 2t+1]}$ in committee ${\sf C}^{e}$. This full key share $F(i,y)$ belongs to user $\pk'$, and $t$+1 of these full key shares can recover $\alpha$.
	 
\end{itemize}

Now we show the generation of an asymmetric bivariate polynomial $f'(x, y)$ with dimension $(t, 2t)$ such that $f'(0,0) = 0$, which is used to update the reduced key shares $f(x, j)$ during proactivization. We denote a subset of ${\sf C}^e$ as $\calU'$, which includes $2t$+1 users. The generation of $f'(x,y)$ requires two steps: univariate zero share, and bivariate zero share. 

\begin{itemize}
	\item{\it Univariate Zero Share.} It requires each user in $\calU'$ to generate a key share $f_j'(y)$ from a common univariate polynomial with dimension $2t$. First, each user $i$ generates a univariate polynomial $f_i'(y) = 0 + a_{1}'y + a_{2}'y^2 + \cdots + a_{2t}'y^{2t}$, and broadcasts it to all users in $\calU'$. Second, each user in $\calU'$ generates a common univariate polynomial $f'(y) = \sum_{i \in [1, 2t+1]} f_i'(y)$ by combining all received polynomials, and obtains a key share $f_j'(y)$ from $f'(y)$.
	
	\item{\it Bivariate Zero Share.} It requires each user in committee ${\sf C}^e$ to generate a key share $f'(x, j)$ from a common bivariate polynomial with dimension $(t, 2t)$. First, each user in $\calU'$ generates a set of reduced shares $\{ f'(i , y ) \}_{i \in [1, n_0]}$ with dimension $t$ from its key share $f_j'(y)$ (i.e., resharing process), where $f'(i, y) = 0 + a_{1,0}'y + a_{2,0}'y^2 + \cdots + a_{2t,0}'y^{2t}$. Since the reduced shares are distributed to all users in committee ${\sf C}^e$, a common bivariate polynomial with dimension $(t, 2t)$ is established: $f'(x, y) = 0 + a_{0,1}'x + a_{1,0}'y + a_{1,1}'xy + a_{1,2}'xy^2 + \cdots + a_{t, 2t}'x^ty^{2t}$. Second, each user in committee ${\sf C}^e$ obtains a reduced key share $f'(x, j)$ by interpolating the received shares $\{ f'(i,j) \}_{j \in [1, 2t+1]}$. The key share $f'(x, j) = 0 + a_{0,1}'x + a_{0,2}'x^2 + \cdots + a_{0, t}'x^{t}$ is used to update $f(x, j)$ in the proactivization. 
	
\end{itemize}

The asymmetric bivariate polynomial $f'(x,y)$ can be reused in another proactivization when sharing secret $\theta$. In other words, multiple handoff protocols with respect to different secrets can be updated using the same bivariate polynomial, with the condition that these handoff protocols are executed within the same committee. 

\smallskip\noindent\textbf{Correctness of {\sf Judge} algorithm.} We show the judge process in detail. First, any public user verifies the connection between a transaction and its modified version, and this connection is publicly verifiable. For example, given two chameleon hash outputs: $( m, m', b, p', C, C', c, c', epk, \sigma, epk', \sigma' )$, the public performs the following operations
		
		\begin{itemize}
		\item verify chameleon hash $ b = g^m \cdot p'^{\sf r} = g^m \cdot p'^{\sf r'} $.
		
		\item verify message-signature pair $(c, \sigma)$ under $(epk, \pk)$, and message-signature pair $(c', \sigma')$ under $(epk', \pk')$.
			
		\item verify $\pk' = \pk \cdot \Delta (\sk)$, where $ \Delta( \sk ) = c'/c = g^{ \sk' - \sk } $ (the meaning of $\Delta(\sk)$ is referred to Section \ref{homo_sig}). Note that $(c, c')$ are derived from the same chameleon trapdoor ${\sf R}$. 
		\end{itemize}

Second, any public user obtains a set of accused committees from interacting with an access blackbox $\calO$. We note that the modifier's delegated secret keys are disallowed to be used in generating $\calO$. If some modifiers use their delegated secret keys to generate $\calO$, the public cannot identify the accused committees correctly because the delegated secret keys may share the same committee index. We argue that it is challenging to prevent such malicious modifiers from using their delegated secret keys to generate $\calO$, as some ABET schemes \cite{liu2013blackbox,lai2018making} (including our proposed one) support a delegation process. The delegation allows a user to obtain some delegated decryption keys by re-randomizing the given decryption key. 

Eventually, the public outputs a transaction-committee pair $(T', {\sf C}^e)$. In particular, we rely on the KZG commitment and PoW consensus to hold a modifier $\pk'$ accountable in an accused committee ${\sf C}^e$. Now we provide more details. 

\begin{itemize}
	\item{\it Share Reduction.} We require user $\pk'$ in committee ${\sf C}^e$ to generate a commitment $C_{f(x,j)}$, which is a KZG commitment to the reduced key shares $\{ f(i,j) \}_{j \in [1, 2t+1] }$, and a set of witnesses $\{ w_{f(i,j)} \}_{j \in [1, 2t+1]}$. A witness $w_{f(i,j)}$ means the witness to evaluation of $f(x,j)$ at $i$. Note that $i \in [1, 2t+1]$ indicates the order of user $\pk'$'s public key in committee ${\sf C}^e$ (we order nodes lexicographically by users' public keys and choose the first $2t+1$). 
	
	\item{\it Full-share Distribution.} We require user $\pk'$ in committee ${\sf C}^e$ to generate a commitment $C_{F(x,j)}$, which is a KZG commitment to the reduced key shares $\{ F(i,j) \}_{j \in [1, 2t+1]}$, and a set of witnesses $w_{F(i,j)}$. A witness $w_{F(i,j)}$ means the witness to evaluation of $F(x,j)$ at $i$. 
	
	\item{\it PoW Consensus.} We require user $\pk'$ to hash the KZG commitment and the set of witnesses, store them to an immutable transaction, and put them on-chain for PoW consensus.
	
\end{itemize}

Overall, the commitment and witness can ensure the correctness of handoff. Specifically, new committee members can verify the correctness of reduced shares from old committee members, thus the correctness of dimension-switching. Meanwhile, the proof of correctness is publicly verifiable, such that any public user can verify that $f(i, j)$ (or $F(i, j)$) is the correct evaluation at $i$ (i.e., user $\pk'$) of the polynomial committed by $C_{f(x, j)}$ (or $C_{F(x, j)}$) in committee ${\sf C}^e$. 

\section{Implementation and Evaluation}
\label{implementation}
In this section, we evaluate the performance of the proposed solution based on a proof-of-concept implementation in Python and Flask framework \cite{flask}. We create a mutable open blockchain system with basic functionalities and a PoW consensus mechanism. The simulated open blockchain system is ``healthy", satisfying the properties of persistence and liveness \cite{garay2015bitcoin}. The system is specifically designed to include ten blocks, each block includes 100 transactions. Please note, that our implementation can easily extend it to real-world applications such as a block containing 3500 transactions. We simulate ten nodes in a peer-to-peer network, each of them is implemented as a lightweight blockchain node. They can also be regarded as the users in a committee. A chain of blocks is established with PoW mechanism by consolidating transactions broadcast by the ten nodes. The implementation code is available on GitHub \cite{sourcecode}.

First, if users append mutable transactions to blockchain, they use the proposed solution to hash the registered message $m$. Later, a miner uses the conventional hash function SHA-256 $\h$ to hash the chameleon hash output $h$ and validates $\h(h)$ using a Merkle tree. Note that the non-hashed components such as randomness ${\sf r}$, are parts of a mutable transaction $T = ( \pk_{\sf CH}, m, h, {\sf r}, C, c, \sigma)$. As a consequence, a modifier can replace $T$ by $T' = ( \pk_{\sf CH}, m', h, {\sf r'}, C', c', \sigma')$ without changing the hash output $\h(h)$.

Second, we mimic a dynamic committee that includes five users, we split the master secret key into five key shares so that each user in a committee holds a key share. We simulate the basic functionality of DPSS, including resharing and updating key shares. Any user can join in or leave from a committee by transmitting those key shares between committee members. In particular, we simulate three users in a committee can collaboratively recover the master secret key and grant access privileges to the modifiers. 

Now, we conclude that: 1) The proposed solution incurs no overhead to chain validation. This is because, rewrite the message in $T$ has no effect on the PoW mechanism, as the chameleon hash output $h$ is used for computing the transaction hash for Merkle tree leaves. 2) The proof-of-concept implementation indicates that the proposed solution can act as an additional layer on top of any open blockchains to perform accountable rewriting. Specifically, we append mutable transactions using the proposed solution to the blockchain, and we allow dynamic committees to grant access privileges for rewriting those mutable transactions. 3) The proposed solution is compatible with existing blockchain systems for the following reasons. The only change which CH-based approach requires is to replace the standard the hash function $\h$ by CH for generating a chameleon hash value $h$ before validating the transactions in each block \cite{ateniese2017redactable}. Besides, DPSS is designed for open blockchains, and decentralized systems \cite{maram2019churp}.

\subsection{Evaluation}

We implement our proposed solution using the Charm framework \cite{akinyele2013charm} and evaluate its performance on a PC with Intel Core i5 (2.7GHz$\times$2) and 8GB RAM. We use Multiple Precision Arithmetic Library, Pairing-Based Cryptography (PBC) Library, and we choose MNT224 curve for pairing, which is the best Type-III paring in PBC. We instantiate the hash function and the pseudo-random generator with the corresponding standard interfaces provided by the Charm framework. 

First, the Setup algorithm takes about 0.52 seconds (s). The running time of KeyGen, Hash, and Adapt algorithms are measured and shown in Figure \ref{runtimes} (a-c). The performance of these algorithms is linear to the number of attributes or the size of policies. Specifically, the run-time of KeyGen takes only 2.37s, even if the size of the policy is 100. We discover that the dominating operation is parsing the access policy ({\bf M}, $\pi$), and we argue that a better designed parsing interface can reduce the overall cost of KeyGen. Moreover, it only takes 2.44s and 3.87s respectively for Hash and Adapt to handle 100 attributes. The run-time cost of such algorithms mainly comes from processing the attributes list and access policy (M, $\pi$), i.e., the coefficient calculation of every attribute and the cost of determining whether a given attribute set satisfies the access policy. 

Second, we evaluate the running time of a $t$-out-of-$n_0$ DPSS protocol, where $n_0$ indicates the number of users in a committee and $t$ is the threshold. Let $t < n_0/2$ be a safe threshold. The overhead includes the distribution cost between committee members, and the polynomial calculation cost. Figure \ref{runtimes} (d) shows that the overall running time is linear to the square number of users $n^2$ in a committee, since more shares need to be distributed and more polynomials need to calculated among $n_0$ users. Our implementation can scale up to larger committees. For example, the running times for $t = 20$ and $t = 30$ are about 3.09s and 7.34s, respectively.

To conclude, the implementation performs the resharing twice and updating once regarding two shared secrets. Besides, the number of updating process is constant in a committee, independent of the number of shared secrets used in ABET. Since only two secrets are needed to be shared and recovered, we argue that the proposed ABET scheme is the most practical one. It is suitable for decentralized systems when applying DPSS to the proposed ABET scheme. On the security-front, because every committee has at most $ \frac{n_0}{3} $ malicious members \cite{luu2016secure} and $\frac{n_0}{2}$+1 committee members recover the shared secrets \cite{maram2019churp}, the malicious committee members cannot dictate the committee and control the rewriting privileges. For the storage cost, we mention that the number of mutable transactions ranges from 2$\%$ to 10$\%$ inside a block \cite{deuber2019redactable}. Each mutable transaction needs to store $T = (\pk_{\sf CH}, m, h, {\sf r}, C, c, \sigma)$. The storage cost of a mutable transaction includes: 1) $2\calL_{\Z} + 3\calL_{\mathbb{G}}$ regarding DL-based chameleon hash; 2) $ \calL_{\Z} + |\delta| \times \calL_{ \mathbb{G}} + (|\delta|+4) \times \calL_{ \mathbb{H}} $ regarding ABET; 3) $\calL_{\Z} + 2\calL_{\mathbb{G}}$ regarding digital signature. The committee's on-chain storage cost regarding DPSS \cite{maram2019churp} is $2(t+1)\times [ \calL_{ \mathbb{G}} + (2t+1) ( \calL_{\Z} + \calL_{ \mathbb{H}}) ] $.

\begin{figure*}

    \subfloat
    {\input{keygen}}
    \subfloat
    {\input{hash}}
    \subfloat
    {\input{adapt}}
    \subfloat
    {\input{dpss}}    
\centering\caption{Run-time of {\sf KeyGen}, {\sf Hash}, {\sf Adapt} algorithms, and {\sf DPSS} scheme.}
\label{runtimes}
\end{figure*}
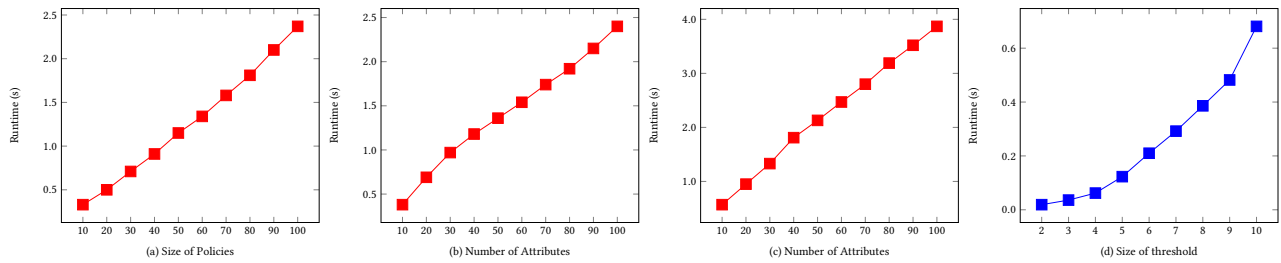

\section{Related Work}
\label{related}

\noindent\textbf{Blockchain Rewriting.} Blockchain rewriting was first introduced by Ateniese et al. \cite{ateniese2017redactable}. They propose to replace the regular SHA256 hash function with a chameleon hash (CH) in blockchain generation \cite{DBLP:conf/ndss/KrawczykR00}. The hashing of CH is parametrized by a public key $\pk$, and CH behaves like a collision-resistant hash function if the chameleon secret key $\sk$ (or trapdoor) is unknown. A trapdoor holder (or modifier) can find collisions and output a new message-randomness pair without changing the hash value.

Camenisch et al. \cite{camenisch2017chameleon} introduced a new cryptographic primitive: chameleon hash with ephemeral trapdoor (CHET). CHET requires that a modifier must have two trapdoors to find collisions: one  trapdoor $\sk$ is associated with the public key $\pk$; the other one is an ephemeral trapdoor $etd$ chosen by the party who initially computed the hash value. CHET provides more control in rewriting in the sense that the party, who computed the hash value, can decide whether the holder of $\sk$ shall be able to rewrite the hash by providing or withholding the ephemeral trapdoor $etd$. 

Derler et al. \cite{DBLP:conf/ndss/DerlerSSS19} proposed policy-based chameleon hash (PCH) to achieve fine-grained rewriting in the blockchain. The proposed PCH replaces the public key encryption scheme in CHET by a ciphertext-policy ABE scheme, such that a modifier must satisfy a policy to find collisions given a hash value. Recently, Tian et al. proposed an accountable PCH for blockchain rewriting (PCHBA) \cite{tian2020policy}. The proposed PCHBA enables the modifiers of transactions to be held accountable for the modified transactions. In particular, PCHBA allows a third party (e.g., key generation center) to resolve any dispute over modified transactions.

In another work, Puddu et al. \cite{puddu2017muchain} proposed $\mu$chain: a mutable blockchain. A transaction owner introduces a set of transactions, including an active transaction and multiple inactive transactions, where the inactive transactions are possible versions of the transaction data (namely, mutations) encrypted by the transaction owner, and the decryption keys are distributed among miners using Shamir's SSS \cite{shamir1979share}. The transaction owner enforces access control policies to define who is allowed to trigger mutations in which context. Upon receiving a mutation-trigger request, a set of miners runs a Multi Party Computation (MPC) protocol to recover the decryption key, decrypt the appropriate version of the transaction and publish it as an active transaction. $\mu$chain incurs considerable overhead due to the use of MPC protocols across multiple miners. It works at both permissioned and permissionless blockchains.

Deuber et al.~\cite{deuber2019redactable} introduced an efficient redactable blockchain in the permissionless setting. The proposed protocol relies on a consensus-based e-voting system \cite{kohno2004analysis}, such that the modification is executed in the chain if a modification request from any public user gathers enough votes from miners (we call it V-CH for convenience). In a follow-up work, Thyagarajan et al. \cite{thyagarajan2020reparo} introduced a protocol called Reparo to repair blockchains, which acts as a publicly verifiable layer on top of any permissionless blockchain. The unique feature of Reparo is that it is immediately integrable into open blockchains in a backward compatible fashion (i.e., any existing blockchains already containing illicit contents can be redacted using Reparo).

There are mainly two types of blockchain rewritings in the literature: CH-based \cite{ateniese2017redactable,camenisch2017chameleon,DBLP:conf/ndss/DerlerSSS19,tian2020policy,derler2020bringing}, and non CH-based \cite{puddu2017muchain,deuber2019redactable,thyagarajan2020reparo}. CH-based blockchain rewritings allow one or more trusted modifiers to rewrite blockchain. The non-CH-based solution requires a threshold number of parties (or miners) to rewrite the blockchain. We stress that both of them aim to rewrite blockchains securely and efficiently. One can apply both of them to redactable blockchains.

Table \ref{compare} shows a comparison between blockchain rewriting related solutions. In this work, we use chameleon hash cryptographic primitive to secure the blockchain rewriting. Our proposed solution supports a fine-grained and controlled rewriting for open blockchains. It holds both the modifiers' public keys and the rewriting privileges accountable for the modified transactions. Overall, this work can be viewed as a step forward from PCH and PCHBA.

\begin{table}

\caption{The comparison between various blockchain rewriting solutions. CH-based indicates CH-based blockchain rewriting. Fine-grained means that each mutable transaction is associated with an access policy such that the transaction can be modified by anyone whose rewriting privilege satisfy the policy.}
\label{compare}
\centering
\scriptsize
\fbox{\parbox{0.42\textwidth}{
\begin{tabular}{l@{\hspace{0.15cm}}c@{\hspace{0.15cm}}c@{\hspace{0.15cm}}c@{\hspace{0.15cm}}c@{\hspace{0.15cm}}c@{\hspace{0.15cm}}c@{\hspace{0.15cm}}c@{\hspace{0.15cm}}c@{\hspace{0.15cm}}c@{\hspace{0.15cm}}}
 
 & CH \cite{ateniese2017redactable} & $\mu$chain \cite{puddu2017muchain} & PCH \cite{DBLP:conf/ndss/DerlerSSS19} & V-CH \cite{deuber2019redactable} & PCHBA \cite{tian2020policy} & Ours \\
\hline

CH-based & \checkmark & $\times$ & \checkmark & $\times$ & \checkmark & \checkmark  \\

Permissionless & \checkmark & \checkmark & $\times$ & \checkmark & $\times$ & \checkmark \\

Fine-grained & $\times$ & $\times$ & \checkmark & $\times$ & \checkmark & \checkmark \\

Accountability & \checkmark & \checkmark & $\times$ & \checkmark  & \checkmark & \checkmark \\

\end{tabular}
}}

\end{table}

\smallskip\noindent\textbf{Dynamic Proactive Secret Sharing.} Proactive security was first introduced by Ostrovsky and Yung \cite{ostrovsky1991withstand}, which is refreshing secrets to withstand compromise. Later, Herzberg et al. \cite{herzberg1995proactive} introduced proactive secret sharing (PSS). The PSS allows the distributed key shares in a SSS to be updated periodically, so that the secret remains secure even if an attacker compromises a threshold number of shareholders in each epoch. However, it did not support dynamic committees because users may join in or leave from a committee dynamically. Desmedt and Jajodia \cite{desmedt1997redistributing} introduced a scheme that redistributes secret shares to new access structure (or new committee). Specifically, a resharing technique is used to change the committee and threshold in PSS. However, the scheme is not verifiable, which disallows PSS to identify the faulty (or malicious) users. The property of verifiability is essential to PSS (i.e., verifiable secret sharing such as Feldman \cite{feldman1987practical}), which holds malicious users accountable. So, the dynamic proactive secret sharing (DPSS) we considered in this work includes verifiability. 

There exist several DPSS schemes in the literature. Wong et al. \cite{wong2002verifiable} introduced a verifiable secret redistribution protocol that supports dynamic committee. The proposed protocol allows new shareholders to verify the validity of their shares after redistribution between different committees. Zhou et al. \cite{zhou2005apss} introduced an APSS, a PSS protocol for asynchronous systems that tolerate denial-of-service attacks. Schultz et al. \cite{schultz2008mobile} introduced a resharing protocol called MPSS. The MPSS supports mobility, which means the group of shareholders can change during resharing. Baron et al. \cite{baron2015communication} introduced a DPSS protocol that achieves a constant amortized communication overhead per secret share. In CCS'19, Maram et al. \cite{maram2019churp} presented a practical DPSS: CHURP. CHURP is designed for open blockchains, and it has very low communication overhead per epoch compared to the existing schemes \cite{wong2002verifiable,zhou2005apss,schultz2008mobile,baron2015communication}. Specifically, the total number of bits transmitted between all committee members in an epoch is substantially lower than in existing schemes. Recently, Benhamouda et al. \cite{benhamouda2020can} introduced anonymous secret redistribution. The benefit is to ensure sharing and resharing of secrets among small dynamic committees.

DPSS can be used to secure blockchain rewriting, such as $\mu$chain \cite{puddu2017muchain}. $\mu$chain relies on encryption with secret sharing (ESS) to hide illegal content, as certain use-cases aim to prevent distribution of illegal content (e.g., child pornography) via the blockchain. ESS allows all the mutable transactions containing illegal content to be encrypted using transaction-specific keys. The transaction-specific keys are split into shares using DPSS \cite{baron2015communication}, and these resulting shares are distributed to a number of miners, which then reshare the keys among all online miners dynamically. In this work, we use KP-ABE with DPSS to ensure blockchain rewiring with fine-grained access control. The master secret key in KP-ABE is split into key shares, and these key shares are distributed to all users in a committee. The key shares can be securely redistributed across dynamic committees. To the best of our knowledge, ours is the first attempt to distribute the master secret key in KP-ABE for decentralized systems. 

\section{Conclusion}

In this paper, we proposed a new framework of accountable fine-grained blockchain rewriting. The proposed framework is designed for open blockchains that require no trust assumptions. Besides, the proposed framework achieves public accountability, which can thwart the malicious rewriting of blockchain. Specifically, public accountability allows the modifiers' public keys and the rewriting privileges to be held accountable for the modified transactions. We presented a practical instantiation, and showed that the proposed solution is suitable for open blockchain applications. In particular, the proof-of-concept implementation demonstrated that our proposed solution can be easily integrated into the existing open blockchains.

\section{Acknowledgments}

This work was supported by the Ministry of Education, Singapore, under its MOE AcRF Tier 2 grant (MOE2018-T2-1-111). Yingjiu Li was supported in part by the Ripple University Blockchain Research Initiative. 

\bibliographystyle{ACM-Reference-Format}
\bibliography{paper}

\appendix

\section{Security Analysis of New Assumption}

\begin{theorem} \label{theo_q}
Let $(\epsilon_1, \epsilon_2, \epsilon_T) : \Z  \rightarrow \{ 0, 1 \}^* $ be three random encodings (injective functions) where $\Z$ is a prime field. $\epsilon_1$ maps all $a \in \Z$ to the string representation $\epsilon_1(g^a)$ of $g^a \in \mathbb{G}$. Similarly, $\epsilon_2$ for $\mathbb{H}$ and $\epsilon_T$ for $\mathbb{G_T}$. If $(a, b ,c, d, \{ z_i \}_{i \in [1, q']} ) \xleftarrow{\text{R}} \Z$ and encodings $\epsilon_1, \epsilon_2, \epsilon_T$ are randomly chosen, then we define the advantage of the adversary in solving the $q'$-type with at most $\calQ$ queries to the group operation oracles $\calO_1, \calO_2, \calO_T$ and the bilinear pairing $\e$ as
\begin{eqnarray*}
 \adv_{\calA}^{q'\text{-}type} (\lambda) &=& | \prob [ \calA (q, \epsilon_1 (1), \epsilon_1 (a), \epsilon_1 (c), \epsilon_1 ((ac)^2), \\ 
& & \epsilon_1 (abd), \epsilon_1 (d/ab), \epsilon_1 (z_i), \epsilon_1 (acz_i), \\
& & \epsilon_1 (ac/z_i), \epsilon_1 (a^2cz_i), \epsilon_1 (b/z_i^2), \\
& & \epsilon_1 (b^2/z_i^2),  \epsilon_1 (acz_i/z_j), \epsilon_1 (bz_i/z_j^2), \\
& & \epsilon_1 (abcz_i/z_j), \epsilon_1 ((ac)^2z_i/z_j), \\
& & \epsilon_2 (1), \epsilon_2 (b), \epsilon_2 (abd), \epsilon_2 (abcd), \\
& & \epsilon_2 (d/ab), \epsilon_2 (c), \epsilon_2 (cd/ab), \\
& = & b: (a,b,c,d, \{ z_i \}_{i \in [1, q']} , s \xleftarrow{\text{R}} \Z, b \in (0,1), \\
& & t_b = abc, t_{1-b} = s ) ] \\
& & - 1/2 | \leq  \frac{ 16 ( \calQ + q' + 22 )^2}{q}
\end{eqnarray*}

\end{theorem}

\begin{proof}
	Let $\calS$ play the following game for $\calA$. $\calS$ maintains three polynomial sized dynamic lists: $ L_1=\{(p_i, \epsilon_{1,i}) \}, L_2=\{ (q_i, \epsilon_{2,i}) \}, L_T=\{ (t_i, \epsilon_{T,i}) \}$, the $p_i \in \Z [A, B, C, D, Z_i, Z_j, T_0, T_1]$ are 8-variate polynomials over $\Z$ (note that $i \neq j$), such that $p_0=1, p_1=A, p_2= C, p_3 = (AC)^2, p_4= ABD, p_5=D/AB, p_6= Z_i, p_7 = ACZ_i, p_8 = AC/Z_i, p_9 = A^2CZ_i, p_{10} = B/Z_i^2, p_{11} = B^2/Z_i^2, p_{12} = ACZ_i/Z_j, p_{13} = BZ_i/Z_j^2, p_{14} = ABCZ_i/Z_j, p_{15} = (AC)^2Z_i/Z_j, q_0=1, q_1 = B, q_2 = ABD, q_3 = ABCD, q_4 = D/AB, q_5 = C, q_6 = CD/AB, p_{16} = T_0, p_{17} = T_1, t_0 =1$, and $(\{\epsilon_{1,i}\}_{i=0}^{16} \in \{ 0, 1 \}^{*} , \{\epsilon_{2, i}\}_{i=0}^5 \in \{ 0, 1 \}^{*}, \{\epsilon_{T, 0}\} \in \{ 0, 1 \}^{*} )$ are arbitrary distinct strings. Therefore, the three lists are initialized as $L_1=\{ (p_i, \epsilon_{1,i}) \}_{i=0}^{17}, L_2= \{ (q_i , \epsilon_{2,i} ) \}_{i=0}^6, L_T= (t_0 , \epsilon_{T,0}) $.

At the beginning of the game, $\calS$ sends the encoding strings $(\{ \epsilon_{1,i} \}_{i=0, \cdots , 17},  \{ \epsilon_{2,i} \}_{i=0, \cdots , 6}, \epsilon_{T, 0} )$ to $\calA$, which includes $q'$+26 strings. Note that the number of encoding string $\epsilon_{1,i}$ is linear to the parameter $q'$. After this, $\calS$ simulates the group operation oracles $\calO_1, \calO_2, \calO_T$ and the bilinear pairing $\e$. We assume that all requested operands are obtained from $\mathcal{S}$.

\begin{itemize}
\item{$\calO_1$:} The group operation involves two operands $\epsilon_{1,i}, \epsilon_{1,j}$. Based on these operands, $\calS$ searches the list $L_1$ for the corresponding polynomials $p_i$ and $p_j$. Then $\calS$ performs the polynomial addition or subtraction $p_l=p_i \pm p_j$ depending on whether multiplication or division is requested. If $p_l$ is in the list $L_1$, then $\calS$ returns the corresponding $\epsilon_l$ to $\calA$. Otherwise, $\calS$ uniformly chooses $\epsilon_{1, l} \in \{ 0, 1 \}^*$, where $\epsilon_{1, l}$ is unique in the encoding string $L_1$, and appends the pair $(p_l, \epsilon_{1, l})$ into the list $L_1$. Finally, $\calS$ returns $\epsilon_{1, l}$ to $\calA$ as the answer. Group operation queries in $\calO_2, \calO_T$ are treated similarly. 

\item{$\e$:} The group operation involves two operands $\epsilon_{T,i}, \epsilon_{T,j}$. Based on these operands, $\calS$ searches the list $L_T$ for the corresponding polynomials $t_i$ and $t_j$. Then $\calS$ performs the polynomial multiplication $t_l=t_i \cdot t_j$. If $t_l$ is in the list $L_T$, then $\calS$ returns the corresponding $\epsilon_{T,l}$ to $\calA$. Otherwise, $\calS$ uniformly chooses $\epsilon_{T, l} \in \{ 0, 1 \}^*$, where $\epsilon_{T, l}$ is unique in the encoding string $L_T$, and appends the pair $(t_l, \epsilon_{T, l})$ into the list $L_T$. Finally, $\calS$ returns $\epsilon_{T, l}$ to $\calA$ as the answer. 

\end{itemize}

After querying at most $\calQ$ times of corresponding oracles, $\calA$ terminates and outputs a guess $b' = \{ 0,1 \}$. At this point, $\calS$ chooses random $a, b, c, d, z_i, z_j, s \in \Z$ and $ t_b = abc $ and $t_{1-b} = s $. $\calS$ sets $A = a, B = b, C=c, D=d, Z_i = z_i, Z_j = z_j, T_0 = t_b, T_1 = t_{1-b}$. The simulation by $\calS$ is perfect (and reveal nothing to $\calA$ about $b$) unless the \abort~event happens. Thus, we bound the probability of event \abort~by analyzing the following cases:

\begin{enumerate}

\item $p_i(a, b, c, d, z_i, z_j, t_0, t_1) = p_j (a, b, c, d, z_i, z_j, t_0, t_1)$: The polynomial $p_i \neq p_j$ due to the construction method of $L_1$, and $(p_i - p_j) (a, b, c, d, z_i, z_j,t_0, t_1)$ is a non-zero polynomial of degree $[0, 6]$, or $q$-2 ($q$-2 is produced by $Z_j^{q-2}$). Since $Z_j \cdot Z_j^{q-2} = Z_j^{q-1} \equiv 1 (\mod q)$, we have $(AC)^2Z_iZ_j \cdot Z_j^{q-2} \equiv (AC)^2Z_iZ_j (\mod q)$.  By using Lemma 1 in \cite{shoup1997lower}, we have $\Pr[(p_i - p_j)(a, b, c, d, z_i, z_j, t_0, t_1) = 0] \leq \frac{6}{q}$ because the maximum degree of $(AC)^2Z_i/Z_j (p_i - p_j) (a, b, c, d, z_i, z_j, t_0, t_1) $ is 6. So, we have $\Pr[p_i(a, b, c, d, z_i, z_j, t_0, t_1) = p_j(a, b, c, d, z_i, z_j, t_0, t_1)] \leq \frac{6}{q}$, and the \abort~probability is $\Pr[\textsf{abort}_1] \leq \frac{6}{q}$.

\item $q_i(a, b, c, d, z_i, z_j, t_0, t_1) = q_j (a, b, c, d, z_i, z_j, t_0, t_1)$: The polynomial $q_i \neq q_j$ due to the construction method of $L_2$, and $(q_i - q_j) (a, b, c, d, z_i, z_j, t_0, t_1)$ is a non-zero polynomial of degree $[0, 4]$, or $q$-2 ($q$-2 is produced by $(AB)^{q-2}$). Since $AB \cdot (AB)^{q-2} = (AB)^{q-1} \equiv 1 (\mod q)$, we have $CDAB \cdot (AB)^{q-2} \equiv CDAB (\mod q)$. The maximum degree of $ CD/AB (q_i - q_j) (a, b, c, d, z_i, z_j, t_0, t_1) $ is 4, so the \abort~probability is $\Pr[\textsf{abort}_2] \leq \frac{4}{q}$.

\item $t_i(a, b, c, d, z_i, z_j, t_0, t_1) = t_j (a, b, c, d, z_i, z_j, t_0, t_1)$: The polynomial $p_i \neq p_j$ due to the construction method of $L_1$, and $(p_i - p_j) (a, b, c, d, z_i, z_j,t_0, t_1)$ is a non-zero polynomial of degree $[0, 6]$, or $q$-2. Since $(AC)^2Z_i \cdot Z_j^{q-2} (t_i - t_j) (a, b, c, d, z_i, z_j, t_0, t_1)$ has degree 6, we have $\Pr[(p_i - p_j)(a, b, c, d, z_i, z_j, t_0, t_1) = 0] \leq \frac{6}{q}$. The \abort~probability is $\Pr[\textsf{abort}_3] \leq \frac{6}{q}$.

\end{enumerate}
   
By summing over all valid pairs $(i,j)$ in each case (i.e., at most  $\binom{ \calQ_{\epsilon_1} + 18 }{2} + \binom{\calQ_{\epsilon_2}+ 7 }{2} + \binom{\calQ_{\epsilon_T}+1}{2}$ pairs), and $\calQ_{\epsilon_1} + \calQ_{\epsilon_2} + \calQ_{\epsilon_T} = \calQ + q' + 26$, we have the \textsf{abort} probability is
\begin{align*}
        \Pr[\textsf{abort}]
        & = \Pr[\textsf{abort}_1] + \Pr[\textsf{abort}_2] + \Pr[\textsf{abort}_3] \\
        & \leq [ \binom{ \calQ_{\epsilon_1} + 18 }{2} + \binom{\calQ_{\epsilon_2}+ 7 }{2} + \binom{\calQ_{\epsilon_T}+1}{2} ] \\
        &  \cdot ( \frac{4}{q} + 2 \frac{6}{q} ) \leq \frac{ 16 ( \calQ + q' + 26 )^2}{q}.
\end{align*}
\end{proof}

\section{P2P Communication Technique \cite{maram2019churp}}

P2P channels can be implemented in different ways depending on the deployment environment. In a permissionless setting, establishing a direct off-chain connection between users is undesirable, as it compromises users' anonymity. Revealing network-layer identities is also dangerous, as it may lead to targeted attacks. Here, we list two approaches. The first approach is to use anonymizing overlay networks such as Tor, at the cost of considerable additional setup and engineering complexity. The second approach uses transaction ghosting, a technique for P2P messaging on a blockchain, which is an overlay on existing blockchain infrastructure. The key idea is to overwrite transactions so that they are broadcast but subsequently dropped by the network. Most of these transactions are broadcast for free. 

Now, we use the Ethereum P2P network as an example; a similar technique can apply to other blockchains such as Bitcoin. Suppose Alice creates a transaction $T$ and sends it to network peers. Note that $T$ includes a message (i.e., payload). The network peers add $T$ to their pool of unconfirmed transactions, known as $mempool$. They propagate $T$ so that it can be included in all peers’ views of the $mempool$. $T$ remains in the $mempool$ until a miner includes it in a block, at which point it is removed, and a transaction fee is transferred from Alice to the miner. The key observation is, until $T$ is mined, Alice can overwrite it with another transaction $T'$ (embed empty message). When this happens, $T$ is dropped from the $mempool$. Thus, both $T$ and $T'$ are propagated to all users, but Alice only pays for $T'$. One can see that $T$ is broadcast for free, and we denote $T$ as pending transaction before overwrite. To conclude, transaction ghosting guarantees that a sender and a receiver can efficiently establish an Ethereum P2P channel via pending transactions such as $T$. The core idea is that the sender can transmit messages to the receiver by embedding them in pending transactions.

%\subsection{Security Models}

%We compare our proposed security models with the closely related security models. First, our proposed indistinguishability and adaptive collision-resistance models are derived from the strong indistinguishability and insider collision-resistance models defined in \cite{DBLP:conf/ndss/DerlerSSS19}. In particular, the proposed adaptive collision-resistance model is more strong compared to the models in \cite{camenisch2017chameleon,DBLP:conf/ndss/DerlerSSS19}, as adversary can periodically corrupt a threshold number of shareholders. Second, our proposed accountability model is derived from the sanitizer accountability model defined in \cite{samelin2020policy}. Note that our proposed accountability model grants a judge oracle to attackers, which determines whether or not a modified transaction links to a responsible modifier and committee.

\section{Security Analysis of Generic Framework}

In this section, we present the security analysis of the proposed generic framework, including indistinguishability, adaptive collision-resistance, and accountability. 

\subsection{Proof of Theorem \ref{the_ind}}

\begin{theorem}\label{the_ind}	
	The proposed generic framework is indistinguishable if the {\sf CH} scheme is indistinguishable. 
\end{theorem} 

\begin{proof} The reduction is executed between an adversary $\calA$ and a simulator $\calS$. Assume that $\calA$ activates at most $n(\lambda)$ chameleon hashes. Let $\calS$ denote a distinguisher against {\sf CH}, who is given a chameleon public key $\pk^*$ and a {\sf HashOrAdapt} oracle, aims to break the indistinguishability of {\sf CH}. In particular, $\calS$ is allowed to access the chameleon trapdoor $Dlog(\pk^*)$ \cite{DBLP:conf/ndss/DerlerSSS19}. $\calS$ randomly chooses $g$ $\in [1, n(\lambda)]$ as a guess for the index of the {\sf HashOrAdapt} query. In the $g$-th query, $\calS$'s challenger directly hashes a message $(h, {\sf r}) \leftarrow {\sf Hash} (\pk^*, m)$, instead of calculating the chameleon hash and randomness $(h, {\sf r})$ using {\sf Adapt} algorithm. 

$\calS$ sets up the game for $\calA$ by distributing a master secret key to a group of users in a committee. $\calS$ can honestly generate secret keys for any modifier associated with an access privilege $\Lambda$. If $\calA$ submits a tuple $(m_0, m_1, \delta)$ in the $g$-th query, then $\calS$ first obtains a chameleon hash $(h_b, {\sf r}_b)$ from his {\sf HashOrAdapt} oracle. Then, $\calS$ simulates a message-signature pair $(c, \sigma)$, and a ciphertext $C$ on message $Dlog(\pk^*)$ according to the protocol specification. Eventually, $\calS$ returns $(h_b, {\sf r}_b, C, c, \sigma)$ to $\calA$. $\calS$ outputs whatever $\calA$ outputs. If $\calA$ guesses the random bit correctly, then $\calS$ can break the indistinguishability of {\sf CH}. 
	
\end{proof}	

\subsection{Proof of Theorem \ref{the_cr}}

\begin{theorem}\label{the_cr}	
	The proposed generic framework is adaptively collision-resistant if the {\sf ABET} scheme is semantically secure, the {\sf CH} scheme is collision-resistant, and the {\sf DPSS} scheme has secrecy. 
\end{theorem}

\begin{proof} We define a sequence of games $\mathbb{G}_i$, $i=0, \cdots, 4$ and let $\adv_i^{GF}$ denote the advantage of the adversary in game $\mathbb{G}_i$. Assume that $\calA$ issues at most $n(\lambda)$ queries to ${\sf Hash}$ oracle. 
\begin{itemize}

\item{$\mathbb{G}_0$:} This is the original game for adaptive collision-resistance.

\item{$\mathbb{G}_1$:} This game is identical to game $\mathbb{G}_1$ except that $\calS$ will output a random bit if $\calA$ outputs a correct master secret key when no more than $t$ users in a committee are corrupted, and the setup is honest (or the dealer is honest). The difference between $\mathbb{G}_0$ and $\mathbb{G}_1$ is negligible if {\sf DPSS} scheme has secrecy. 
\begin{equation}
\left|\adv_0^{GF}-\adv_1^{GF}\right| \leq \adv_{\calS}^{\sf DPSS}(\lambda).
\end{equation}

\item{$\mathbb{G}_2$:} This game is identical to game $\mathbb{G}_1$ except the following difference: $\calS$ randomly chooses $g \in [1, n(\lambda)]$ as a guess for the index of the ${\sf Hash'}$ oracle which returns the chameleon hash $(h^*, m^*, {\sf r}^*, C^*, c^*, \sigma^*)$. $\calS$ will output a random bit if $\calA$'s attacking query does not occur in the $g$-th query. Therefore, we have
\begin{equation}
\adv_1^{GF}= n(\lambda) \cdot \adv_2^{GF}
\end{equation}

\item{$\mathbb{G}_3$:} This game is identical to game $\mathbb{G}_2$ except that in the $g$-th query, the encrypted message $\sk_{\sf CH}$ in $C^*$ is replaced by $``\bot"$ (i.e., an empty value). Below we show that the difference between $\mathbb{G}_2$ and $\mathbb{G}_3$ is negligible if {\sf ABET} scheme is semantically secure. 

Let $\calS$ denote an attacker against {\sf ABET} with semantic security, who is given a public key $\pk^*$ and a key generation oracle, aims to distinguish between encryptions of $M_0$ and $M_1$ associated with a challenge index $j^*$ and a challenge set of attributes $\delta^*$, which are predetermined at the beginning of the game for semantic security. $\calS$ simulates the game for $\calA$ as follows. 

\begin{itemize}

\item $\calS$ sets up $\mpk_{\sf ABET} = \pk^*$ and completes the remainder of {\sf Setup} honestly, which includes user's key pairs and chameleon key pairs for hashing in {\sf CH}. $\calS$ returns all public information to $\calA$. 

\item $\calS$ can honestly answer the queries made by $\calA$ regarding decryption keys using his given oracle, such that $\Lambda_i ( \delta^* ) \neq 1, i \neq j^* $. In the $g$-th query, upon receiving a hash query w.r.t., a hashed message $m^*$ from $\calA$. $\calS$ first submits two messages (i.e., $[ M_0 = \sk_{\sf CH}, M_1 = \bot ]$) to his challenger, and obtains a challenge ciphertext $C^*$ under index $j^*$ and $\delta^*$. Then, $\calS$ returns the tuple $(h^*, m^*, {\sf r}^*, C^*, c^*, \sigma^*)$ to $\calA$. Note that $\calS$ can simulate the message-signature pair $(c^*, \sigma^*)$ honestly using user's key pairs. Besides, $\calS$ can simulate the adapt query successfully using $\sk_{\sf CH}$.

\end{itemize}

\noindent If the encrypted message in $C^*$ is $\sk_{\sf CH}$, then the simulation is consistent with $\mathbb{G}_2$; Otherwise, the simulation is consistent with $\mathbb{G}_3$. Therefore, if the advantage of $\calA$ is significantly different in $\mathbb{G}_2$ and $\mathbb{G}_3$, $\calS$ can break the semantic security of the {\sf ABET}. Hence, we have
\begin{equation}
\left|\adv_2^{GF}-\adv_3^{GF}\right| \leq \adv_{\calS}^{\sf ABET}(\lambda).
\end{equation}

\item{$\mathbb{G}_4$:} This game is identical to game $\mathbb{G}_3$ except that in the $g$-th query, $\calS$ outputs a random bit if $\calA$ outputs a valid collision $(h^{*}, m^{*'}, {\sf r}^{*'}, C^{*'} , c^{*'}, \sigma^{*'})$, and it was not previously returned by the {\sf Adapt'} oracle. Below we show that the difference between $\mathbb{G}_3$ and $\mathbb{G}_4$ is negligible if {\sf CH} is collision-resistant. 

Let $\calS$ denote an attacker against {\sf CH} with collision-resistant, who is given a chameleon public key $\pk^*$ and an {\sf Adapt'} oracle, aims to find a collision which was not simulated by the {\sf Adapt'} oracle. $\calS$ simulates the game for $\calA$ as follows. 

\begin{itemize}

\item $\calS$ sets up $\pk_{\sf CH} = \pk^*$ for the $g$-th hash query, and completes the remainder of {\sf Setup} honestly, which includes user's key pairs and master key pair in {\sf ABET}. $\calS$ returns all public information to $\calA$. 

\item $\calS$ can simulate all queries made by $\calA$ except adapt queries. If $\calA$ submits an adapt query in the form of $(h, m, {\sf r}, C, c, \sigma, m')$, then $\calS$ obtains a randomness ${\sf r'}$ from his {\sf Adapt'} oracle, and returns $( h, m', {\sf r'}, C',  c', \sigma' )$ to $\calA$. In particular, $\calS$ simulates the $g$-th hash query as $(h^*, m^*, {\sf r}^*, C^*,  c^*, \sigma^*)$ w.r.t. a hashed message $m^*$, where $C^* \leftarrow {\sf Enc}_{\sf ABET}(\mpk^*, \bot, \delta^*, j^*)$, $c^*$ is derived from $\bot$ because $Dlog(\pk^*)$ is unknown.

\item If $\calA$ outputs a collision $(h^*, m^*, {\sf r}^*, C^*, c^*, \sigma^*, m^{*'}, \break {\sf r}^{*'}, C^{*'}, c^{*'},  \sigma^{*'})$ with respect to the $g$-th query, and all relevant checks are succeed, then $\calS$ output $(h^*, m^{*'}, {\sf r}^{*'}, C^{*'}, c^{*'}, \sigma^{*'})$ as a valid collision to {\sf CH}; Otherwise, $\calS$ aborts the game. Therefore, we have
\begin{equation}
\left|\adv_3^{GF}-\adv_4^{GF}\right| \leq \adv_{\calS}^{\sf CH}(\lambda).
\end{equation}

\end{itemize}
	
\noindent Combining the above results together, we have
\begin{eqnarray*}
\adv_{\mathcal{A}}^{GF}(\lambda) & \le & n(\lambda) \cdot  ( \adv_{\calS}^{\sf DPSS}(\lambda) +  \adv_{\calS}^{\sf ABET}(\lambda) \\
& & + \adv_{\calS}^{\sf CH}(\lambda) ).
\end{eqnarray*}

\end{itemize}
	
\end{proof}

\subsection{Proof of Theorem \ref{the_acc}}

\begin{theorem}\label{the_acc}	
	The proposed generic framework is accountable if the $\Sigma$ scheme is EUF-CMA secure, and the {\sf DPSS} scheme has correctness.  
\end{theorem} 

\begin{proof} We define a sequence of games $\mathbb{G}_i$, $i=0, \cdots, 2$ and let $\adv_i^{GF}$ denote the advantage of the adversary in game $\mathbb{G}_i$. 
	
\begin{itemize}

\item{$\mathbb{G}_0$:} This is the original game for accountability.

\item{$\mathbb{G}_1$:} This game is identical to game $\mathbb{G}_1$ except that $\calS$ will output a random bit if all honest users in a committee outputs a master secret key $\msk'$ such that $\msk' \neq \msk$, and the setup is honest. The difference between $\mathbb{G}_0$ and $\mathbb{G}_1$ is negligible if {\sf DPSS} scheme has correctness. 
\begin{equation}
\left|\adv_0^{GF}-\adv_1^{GF}\right| \leq \adv_{\calS}^{\sf DPSS}(\lambda).
\end{equation}

\item{$\mathbb{G}_2$:} This game is identical to game $\mathbb{G}_1$ except that $\calS$ will output a random bit if $\calA$ outputs a valid forgery $\sigma^*$, where $\sigma^*$ was not previously simulated by $\calS$ and the user is honest. The difference between $\mathbb{G}_1$ and $\mathbb{G}_2$ is negligible if $\Sigma$ is EUF-CMA secure.  

Let $\calF$ denote a forger against $\Sigma$, who is given a public key $\pk^*$ and a signing oracle $\calO^{\sf Sign}$, aims to break the EUF-CMA security of $\Sigma$. Assume that $\calA$ activates at most $n$ users in the system. 

\begin{itemize}
	\item $\calF$ randomly chooses a user in the system and sets up its public key as $\pk^*$. $\calF$ completes the remainder of {\sf Setup} honestly. Below we mainly focus on user $\pk^*$ only. 
	
	\item To simulate a chameleon hash for message $m$, $\calF$ first obtains a signature $\sigma$ from his signing oracle $\calO^{\sf Sign}$. Then, $\calF$ generates chameleon hash and ciphertext honestly because $\calF$ chooses the chameleon secret key $\sk_{\sf CH}$, and returns $(m, h, {\sf r}, C, c, \sigma )$ to $\calA$. Besides, the message-signature pairs and collisions can be perfectly simulated by $\calF$ for any adapt query. $\calF$ records all the simulated message-signature pairs by including them to a set $\calQ$. 
	
	\item When forging attack occurs, i.e., $\calA$ outputs $( m^*, h^*, {\sf r}^*, C^*, c^*, \sigma^* )$, $\calF$ checks whether:

\begin{itemize}

\item the forging attack happens to user $\pk^*$;

\item the ciphertext $C^*$ encrypts the chameleon trapdoor $\sk_{\sf CH}$;

\item the message-signature pair $(c^*, \sigma^*)$ is derived from the chameleon trapdoor $\sk_{\sf CH}$; 

\item the message-signature pair $(c^*, \sigma^{*}) \notin \calQ$;

\item $1 \leftarrow \Sigma.{\sf Verify}( \pk^*, c^*, \sigma^{*} )$ and $1 \leftarrow {\sf Verify}( \pk_{\sf CH}, \break m^*, h^*, {\sf r}^* )$.

\end{itemize}
	
\end{itemize}

\noindent If all the above conditions hold, $\calF$ confirms that it as a successful forgery from $\calA$, then $\calF$ extracts the forgery via $\sigma \leftarrow M_{\Sigma} ( {\sf PP}, \pk^*, \sigma^{*}, \Delta(\sk) ) $ due to the homomorphic property of $\Sigma$ (regarding keys and signatures), where $\Delta(\sk)$ is derived from $(c, c^*)$. To this end, $\calF$ outputs $\sigma$ as its own forgery; Otherwise, $\calF$ aborts the game. 
\begin{equation}
\left|\adv_1^{GF}-\adv_2^{GF}\right| \leq n \cdot \adv_{\calF}^{\Sigma}(\lambda).
\end{equation}

\end{itemize}

\noindent Combining the above results together, we have
\begin{eqnarray*}
\adv_{\mathcal{A}}^{GF}(\lambda) & \le &  \adv_{\calS}^{\sf DPSS}(\lambda) + n \cdot \adv_{\calF}^{\Sigma}(\lambda).
\end{eqnarray*}

\end{proof}

\section{Security Analysis of ABET}

In this section, we present the security analysis of the proposed ABET scheme, including semantic security and ciphertext anonymity. 

\subsection{Semantic Security}

Informally, an ABE scheme is secure against chosen plaintext attacks if no group of colluding users can distinguish between encryption of $M_0$ and $M_1$ under an index and a set of attributes $\delta^*$ of an attacker's choice as long as no member of the group is authorized to decrypt on her own. The selective security is defined as an index, as well as a set of attributes $\delta^*$, are chosen by attackers at the beginning security experiment. The semantic security model here is based on the selective security model defined in \cite{rouselakis2013practical}. 

\begin{theorem}
The proposed ABET scheme is semantically secure in the standard model if the $q'$-type assumption is held in the asymmetric pairing groups. 	
\end{theorem}

\begin{proof} Let $\calS$ denote a $q'$-type problem attacker, who is given the terms from the assumption, aims to distinguish $g^{abc}$ and $g^s$. The reduction is performed as follows. 

\begin{itemize}
	\item $\calS$ simulates master public key $\mpk = (g, u, v, w, h, \e(g, h)^{\alpha}, \{ g_1^{\alpha}, \cdots \break g_k^{\alpha} \},  \{ h_1^{\alpha}, \cdots h_k^{\alpha} \}, g^{\beta}, h^{1/ \alpha}, h^{\beta/\alpha}, \e(g, h)^{\theta /\alpha } )$ as follows: $\e(g, h)^{\alpha} = \e(g^a, h^b)$, $\{ g_1^{\alpha}, \cdots g_k^{\alpha} \} = \{ g^{abdz_1}, \cdots, g^{abdz_k} \}$, $\{ h_1^{\alpha}, \cdots h_k^{\alpha} \} = \{ h^{abdz_1}, \break \cdots, h^{abdz_k} \}$, $h^{1/\alpha} = h^{d/ab}$, $h^{\beta/\alpha} = h^{\beta d/ab}$, $\e(g , h)^{\theta/\alpha} = \e(g^{\theta}, h^{d/ab})$, and $(u,v,w)$ are simulated using the same method described in \cite{rouselakis2013practical}. Note that $( \beta, \theta, \{ z_1, \cdots, z_k \} )$ are randomly chosen by $\calS$, and $\alpha$ (or $1/\alpha$) is implicitly assigned as $ab$ (or $d/ab$) from the $q'$-type assumption. $\calA$ submits a challenge index $j^* = \{ I_1, \cdots, I_j \}$ and a set of attributes $\delta^*$ to $\calS$.
	
	\item $\calS$ simulates decryption keys $\sk_{\Lambda_i} = ( \{ \sk_{\tau} \}_{\tau \in [n_1]}, \sk_0, \sk_1, \sk_2 ) $ (note that $\Lambda_i (\delta^*) \neq 1$) as follows: $\sk_0 = ( g^{dt/ab}, g^r)$, $\sk_1 = g^{\theta} \cdot \widehat{i^t} \cdot g^{\beta \cdot r} $, $\sk_2 = \{ g_{i-1}^{abdt}, \cdots, g_1^{abdt} \}$, where $\widehat{i} = g_k^{abdI_1} \cdots g_i^{abdI_i} \cdot g $.  Note that $\{ \sk_{\tau} \}_{\tau \in [n_1]}$ is simulated using the same method described in \cite{rouselakis2013practical}, and $(t, r)$ are randomly chosen by $\calS$. 
	
	\item $\calS$ simulates challenge ciphertext $C^* = (ct, \{ ct_{\tau, 1}, ct_{\tau, 2} \}_{\tau \in [\delta]}, ct_0, ct_1 )$ as follows: $ct = M_b \oplus \h ( T || \e(g, h)^{\theta cd/ab} )$, $ct_0 = ( h^c, h^{cd/ab}, h^{\beta cd/ab})$, and $ct_1 = \widehat{j^c} = h_k^{abcdI_1} \cdots h_j^{abcdI_j} \cdot h $. Note that $\{ ct_{\tau, 1}, ct_{\tau, 2} \}_{\tau \in [\delta]}$ are simulated using the same method described in \cite{rouselakis2013practical}, $s$ is implicitly assigned as $c$ from the $q'$-type assumption, and $T$ can be either $\e(g, h)^{abc}$ or $\e(g, h)^s$. 
	
\end{itemize}

\noindent Finally, $\calS$ outputs whatever $\calA$ outputs. If $\calA$ guesses the random bit correctly, then $\calS$ can break the $q'$-type problem. 

\end{proof}

\subsection{Ciphertext Anonymity}

Informally, ciphertext anonymity requires that any third party cannot distinguish the encryption of a chosen message for a first chosen index from the encryption of the same message for a second chosen index. In other words, the attacker cannot decide whether a ciphertext was encrypted for a chosen index or a random index. We prove the ABET scheme has selective ciphertext anonymity (i.e., the index is chosen prior to the security experiment).

\begin{theorem}
The proposed ABET scheme is anonymous if the eDDH assumption is held in the asymmetric pairing groups. 	
\end{theorem}

\begin{proof} Let $\calS$ denote an eDDH problem distinguisher, who is given terms from the assumption, aims to distinguish $h^{c/ab}$ and $h^s$. The reduction is performed as follows. 

\begin{itemize}

\item $\calS$ simulates master public key $\mpk = (g, u, v, w, h, \e(g, h)^{\alpha}, \{ g_1^{\alpha}, \cdots \break g_k^{\alpha} \}, \{ h_1^{\alpha}, \cdots  h_k^{\alpha} \}, g^{\beta}, h^{1/ \alpha}, h^{\beta/\alpha}, \e(g, h)^{\theta /\alpha } )$ as follows: $\e(g,h)^{\alpha} = \e(g, h)^{ab}$, $ \{ g_1^{\alpha}, \cdots, g_k^{\alpha} \} = \{ g_1^{ab}, \cdots, g_k^{ab} \} $, $ \{ h_1^{\alpha}, \cdots, h_k^{\alpha} \} = \{ h_1^{ab}, \cdots, \break h_k^{ab} \} $, $h^{1/\alpha} = h^{1/ab}$, $h^{\beta/\alpha} = h^{\beta/ab}$, $\e(g, h)^{\theta/\alpha} = \e(g, h)^{\theta/ab}$. Note that  $\calS$ randomly chooses $(u,v, w)$ and $(\beta, \theta, \{ z_i \})$, and implicitly sets $\alpha = ab$. $\calA$ submits a challenge index $j^* = \{ I_1, \cdots, I_j \}$ to $\calS$. 
 
\item $\calS$ simulates a decryption key $\sk_{\Lambda_i} = ( \{ \sk_{\tau} \}_{\tau \in [n_1]}, \sk_0, \sk_1, \sk_2 ) $ with respect to index $i = \{ I_1, \cdots, I_i \} $ as follows: $ \sk_0 = (g^t, g^{\beta \cdot r} \cdot g^{- \Sigma (z_i I_i) t \cdot b} )$, $\sk_1 = g^{ \theta } \cdot g^{a \cdot r}$, where $(t, r) \in \Z $ are randomly chosen by $\calS$. The components $( \{ \sk_{\tau} \}_{\tau \in [n_1]}, \sk_2 )$ are honestly simulated by $\calS$. The simulated components $( g^t , g^{\beta \cdot r} \cdot g^{- \Sigma (z_i I_i) t\cdot b}, g^{\theta} \cdot g^{a \cdot r}  )$ are correctly distributed, because $g^{\theta} \cdot g^{a \cdot r}  = g^{\theta} \cdot \widehat{i}^{t} \cdot g^{a [ \beta \cdot r- \Sigma (z_i I_i) t \cdot b] } = g^{\theta} \cdot \widehat{i}^{ t} \cdot g^{a \cdot \overline{r}} $, where $\overline{r} =\beta \cdot r - \Sigma (z_i I_i) t \cdot b$, and $\widehat{i} = g_{k}^{ab I_1} \cdots g_i^{ab I_i} \cdot g$. So, the simulated components $( g^{t} , g^{\overline{r}}, g^{\theta} \cdot \widehat{i}^{t} \cdot g^{a \cdot \overline{r}})$ match the real distribution.

\item $\calS$ simulates the challenge ciphertext $C^* = (ct, \{ ct_{\tau, 1}, ct_{\tau, 2} \}_{\tau \in [\delta]}, \break ct_0, ct_1 )$ with respect to index $j^*$ as follows: $ct = M_b \oplus \h (\e(g, h)^{abc} || \break \e(g, T)^{\beta} ) $, $ct_0 = (h^c, T, T^{\beta})$, and $ct_1 = h_k^{abcI_1} \cdots h_j^{abcI_j} \cdot h^c $. Note that $\calS$ simulates $\{ ct_{\tau, 1}, ct_{\tau, 2} \}_{\tau \in [\delta]}$ honestly, and $T$ can be either $h^{c/ab}$ or $h^s$. 

\end{itemize}

\noindent Finally, $\calS$ outputs whatever $\calA$ outputs. If $\calA$ guesses the random bit correctly, then $\calS$ can break the eDDH problem. 
\end{proof}

\end{document}

%% file: mydefs.tex
\usepackage{amsmath,amsfonts}
\usepackage[caption=false]{subfig}

\input{stdeqn}

\newcommand{\h}{{\tt H}}
\newcommand{\e}{{\mathsf{\hat{e}}}}
\newcommand{\mpk}{{\tt mpk}}
\newcommand{\msk}{{\tt msk}}
\newcommand{\sk}{{\tt sk}}
\newcommand{\pk}{{\tt pk}}

\newcommand{\abort}{{\tt abort}}

\newcommand{\iseq}{{\stackrel{?}{=}}}

\newcommand{\adv}{{\mathtt{Adv}}}

%% file: stdeqn.tex
\newcommand{\bdfn}{\begin{dfn}	\rm}
\newcommand{\edfn}{\end{dfn}}
\newcommand{\bthm}{\begin{thm}	\rm}
\newcommand{\ethm}{\end{thm}}
\newcommand{\blem}{\begin{lem}	\rm}
\newcommand{\elem}{\end{lem}}
\newcommand{\bcor}{\begin{cor}	\rm}
\newcommand{\ecor}{\end{cor}}
\newcommand{\bexa}{\begin{exa}	\rm}
\newcommand{\eexa}{\hfill$\Box$\end{exa}}
\newcommand{\bexr}{\begin{exr}	\rm}
\newcommand{\eexr}{\end{exr}}
\newcommand{\bpbm}{\begin{pbm}	\rm}
\newcommand{\epbm}{\end{pbm}}

\newcommand{\Z}{{\mathbb{Z}_q}}

\newcommand{\EQ}{\begin{eqnarray*}}
\newcommand{\EN}{\end{eqnarray*}}
\newcommand{\eq}{\begin{eqnarray}}
\newcommand{\en}{\end{eqnarray}}

\newcommand{\prob}{{\mathrm{Pr}}}

\newcommand{\calA}{{\mathcal{A}}}

\newcommand{\calD}{{\mathcal{D}}}

\newcommand{\calF}{{\mathcal{F}}}

\newcommand{\calL}{{\mathcal{L}}}
\newcommand{\calM}{{\mathcal{M}}}

\newcommand{\calO}{{\mathcal{O}}}

\newcommand{\calQ}{{\mathcal{Q}}}

\newcommand{\calS}{{\mathcal{S}}}

\newcommand{\calU}{{\mathcal{U}}}

%% file: keygen.tex
\begin{tikzpicture}[scale = 0.5,>=stealth]
\begin{axis}[
    legend cell align={left},
    xtick={10, 20,...,100},
    symbolic x coords={10, 20, 30, 40, 50, 60, 70, 80, 90, 100},
    %ytick distance=0.4,
    % grid=major,
    yticklabel style={
            /pgf/number format/.cd,
            fixed,
            fixed zerofill,
            precision=1,
        	/tikz/.cd
    },
    legend style={
    at={(0,0)},
    anchor=north west,at={(axis description cs:0,1.0)}},
    xlabel={(a) Size of Policies},
    ylabel={Runtime (s)},
]
\addplot [mark options={red}, mark=square*, mark size=4pt, style={red}] coordinates {
    (10, 0.33) (20,0.50) (30,0.71) (40,0.91) (50, 1.15) (60, 1.34) (70, 1.58) (80, 1.81) (90, 2.10) (100, 2.37)
};
    % \legend{Super,Method,Argument,Arg. (one-time)}
\end{axis}
\end{tikzpicture}

%% file: hash.tex
\begin{tikzpicture}[scale = 0.5,>=stealth]
\begin{axis}[
    legend cell align={left},
    xtick={10, 20,...,100},
    symbolic x coords={10, 20, 30, 40, 50, 60, 70, 80, 90, 100},
    %ytick distance=0.4,
    % grid=major,
    yticklabel style={
            /pgf/number format/.cd,
            fixed,
            fixed zerofill,
            precision=1,
        	/tikz/.cd
    },
    legend style={
    at={(0,0)},
    anchor=north west,at={(axis description cs:0,1.0)}},
    xlabel={(b) Number of Attributes},
    ylabel={Runtime (s)}
]
\addplot [mark options={red}, mark=square*, mark size=4pt, style={red}] coordinates {
    (10, 0.38) (20,0.69) (30,0.97) (40,1.18) (50, 1.36) (60, 1.54) (70, 1.74) (80, 1.92) (90, 2.15) (100, 2.40)
};
    % \legend{}
\end{axis}
\end{tikzpicture}

%% file: adapt.tex
\begin{tikzpicture}[scale = 0.5,>=stealth]
\begin{axis}[
    legend cell align={left},
    xtick={10, 20,...,100},
    symbolic x coords={10, 20, 30, 40, 50, 60, 70, 80, 90, 100},
    %ytick distance=0.8,
    % grid=major,
    yticklabel style={
            /pgf/number format/.cd,
            fixed,
            fixed zerofill,
            precision=1,
        	/tikz/.cd
    },
    legend style={
    at={(0,0)},
    anchor=north west,at={(axis description cs:0,1.0)}},
    xlabel={(c) Number of Attributes},
    ylabel={Runtime (s)}
]
\addplot [mark options={red}, mark=square*, mark size=4pt, style={red}] coordinates {
    (10, 0.57) (20,0.95) (30,1.33) (40,1.81) (50, 2.13) (60, 2.47) (70, 2.80) (80, 3.19) (90, 3.52) (100, 3.87)
};
    % \legend{}
\end{axis}
\end{tikzpicture}

%% file: dpss.tex
\begin{tikzpicture}[scale = 0.5,>=stealth]
\begin{axis}[
    legend cell align={left},
    xtick={2,3,4,5,6,7,8,9,10},
    symbolic x coords={2,3,4,5,6,7,8,9,10},
    %ytick distance=0.2,
    % grid=major,
    yticklabel style={
            /pgf/number format/.cd,
            fixed,
            fixed zerofill,
            precision=1,
        	/tikz/.cd
    },
    legend style={
    at={(0,0)},
    anchor=north west,at={(axis description cs:0,1.0)}},
    xlabel={(d) Size of threshold},
    ylabel={Runtime (s)},
]
\addplot [mark options={blue}, mark=square*, mark size=4pt, style={blue}] coordinates {
    (2,0.019) (3,0.036) (4,0.062) (5, 0.123) (6, 0.210) (7, 0.292) (8, 0.386) (9, 0.482) (10, 0.681)
};
\end{axis}
\end{tikzpicture}